\documentclass[11pt]{article} %% Default font size: 11-point.

\usepackage[margin=1in]{geometry} %% Page margins: 1$''$. Default letter-size (8.5 x 11 inch) paper.
\usepackage[utf8]{inputenc}
\usepackage[T1]{fontenc}
\usepackage{mathpazo} %% Work with pdflatex.

\usepackage{authblk,comment,xspace}
\usepackage[normalem]{ulem} %% Add normalem to let \emph{text} command has its original effect.

\usepackage{booktabs,colortbl,multicol,multirow,subcaption} %% Small Guide to Making Nice Tables: https://people.inf.ethz.ch/markusp/teaching/guides/guide-tables.pdf

\usepackage{enumitem} %% [label=(\arabic*)]; \roman*; \alph*; \Alph*
\usepackage[dvipsnames]{xcolor} %% Required for specifyiƒng colors by name
\definecolor{ptblue}{RGB}{15,76,129} %% PANTONE 19-4052 Classic Blue in Year 2020.
\definecolor{ptemerald}{HTML}{009473} %% PANTONE 17-5641 TPX Emerald in Year 2013.
\definecolor{ptgray}{HTML}{939597} %% PANTONE 17-5104 Ultimate Gray in Year 2021.
\usepackage{xfrac}
\usepackage{tikz} %% Required for drawing custom shapes

\usepackage{amsmath,amsfonts,amssymb,amsthm,bbm,mathtools,pifont,thm-restate}
\allowdisplaybreaks
\newcommand{\cmark}{{\ding{51}}\xspace}
\newcommand{\xmark}{{\ding{55}}\xspace}

\let\OLDland\land
\renewcommand{\land}{\:\OLDland\:}
\let\OLDlor\lor
\renewcommand{\lor}{\:\OLDlor\:}
\let\OLDforall\forall
\renewcommand{\forall}{\;\OLDforall\:}
\let\OLDexists\exists
\renewcommand{\exists}{\;\OLDexists\,}

\DeclareMathOperator*{\argmax}{arg\,max}

\usepackage[ruled,linesnumbered,vlined]{algorithm2e}
\SetKwRepeat{Do}{do}{while} %% Use ``Do...While...'' instead of ``Repeat...Until...'': See https://tex.stackexchange.com/questions/212301/do-while-loop-in-algorithm2e

\SetCommentSty{mycommentfont}

\usepackage[square]{natbib} %% Work with \bibliographystyle{plainnat}
\usepackage{hyperref}
\hypersetup{
	linktocpage,
	colorlinks=true,
	citecolor=cobalt, %% colour of links to bibliography
	urlcolor=ptblue, %% colour of links to external links
	linkcolor=cobalt, %% colour of links to internal links
}
\definecolor{cobalt}{rgb}{0.0, 0.28, 0.67}
\usepackage{cleveref} %% ``cleveref'' must be loaded after ``hyperref'' and ``amsmath'', but beofre ``theorem-like environments''.

\theoremstyle{plain}
\newtheorem{theorem}{Theorem}[section]

\newtheorem{proposition}[theorem]{Proposition}

\theoremstyle{definition}
\newtheorem{definition}[theorem]{Definition}
\newtheorem{example}[theorem]{Example}

\newtheoremstyle{bfnote}% Remark environment
{}{}%
{}{}%
{\bfseries}{.\hspace{0.25cm}}%
{ }%
{\thmname{#1}\thmnumber{ #2}\thmnote{ (#3)}}
\theoremstyle{bfnote}
\newtheorem{remark}{Remark}

\newcommand{\approvalBallotOfvoter}[1]{A_{#1}} %% A_i is currently used for an integral outcome.
\newcommand{\nash}{$\mathsf{NASH}$\xspace}
\newcommand{\cut}{$\mathsf{CUT}$\xspace}
\newcommand{\fut}{$\mathsf{FUT}$\xspace}
\newcommand{\ues}{$\mathsf{UES}$\xspace}
\newcommand{\map}{$\mathsf{MP}$\xspace}
\newcommand{\mps}[1]{$#1\text{--}\mathsf{MSP}$\xspace}
\title{Approximately Fair and Population Consistent \\Budget Division via Simple Payment Schemes}
\author{Haris Aziz \qquad
Patrick Lederer\thanks{Corresponding author.} \qquad
Xinhang Lu \qquad
Mashbat Suzuki \qquad
Jeremy Vollen \medskip\\
UNSW Sydney, Australia \medskip\\
\nolinkurl{{haris.aziz, p.lederer, xinhang.lu, mashbat.suzuki, j.vollen}@unsw.edu.au}}
\date{}

\usepackage{abstract}
\begin{document}
\maketitle

\begin{abstract}\vspace{-0.2cm}
	In approval-based budget division, a budget needs to be distributed to candidates based on the voters' approval ballots over these candidates. In the pursuit of a simple, consistent, and approximately fair rule for this setting, we introduce the maximum payment rule (\map). Under this rule, each voter controls a part of the budget and, in each step, the corresponding voters allocate their entire budget to the candidate approved by the largest number of voters with non-zero budget. We show that \map meets our criteria as it satisfies monotonicity and a demanding population consistency condition and gives a $2$-approximation to a fairness notion called average fair share (AFS). Moreover, we generalize \map to the class of sequential payment rule and prove that it is the most desirable rule in this class: all sequential payment rules but \map and one other rule fail monotonicity while only allowing for a small improvement in the approximation ratio to AFS. \medskip
	%
%n approval-based budget division, the goal is to distribute a divisible resource such as money or time to the candidates based on the voters' approval ballots over the candidates.
%In this setting, the Nash product rule (\nash) is known as the fairest rule as it satisfies demanding conditions such as average fair share and core.
%However, while \nash is fair, it is a rather convoluted rule and it fails other desirable properties such as monotonicity and, as we will show, ranked population consistency.
%Motivated by these observations, we suggest in this paper the class of sequential payment rules, where each voter is given control over a part of the budget and voters repeatedly spend their share on their approved candidates to determine the final distribution.
%We then show that all sequential payment rules satisfy ranked population consistency, and we identify two particularly appealing rules within this class:
%the maximum payment rule (\map), where voters spend their whole share on their first approved candidate, and the $\frac{1}{3}$-multiplicative sequential payment rule (\mps{\frac{1}{3}}), where the payment willingness of each voter is discounted by~$\frac{1}{3}$ whenever he pays for a candidate. In more detail, we show \emph{(i)} that \map is, apart from one other rule, the only monotonic sequential payment rule and a $2$-approximation to average fair share, and \emph{(ii)} that \mps{\frac{1}{3}} is a $\frac{3}{2}$-approximation to average fair share, which is optimal within the class of sequential payment rules.\medskip
\end{abstract}

\textbf{JEL Classification code:} D7
%% keywords: Social Choice Theory, Budget Division, Fairness, Approval Preferences

\section{Introduction}

Suppose that a city council wants to support the local sport clubs by distributing a fixed amount of money between them. In order to allocate the money in a satisfactory way, the council asks for the citizens' preferences over the sport clubs. However, given these preferences, how should the budget be distributed? This simple example describes a fundamental research problem in the realm of social choice theory, which we will refer to as budget division.\footnote{In the literature, budget division is referred to by a number of names such as randomized social choice \citep{BMS05a}, fair mixing \citep{ABM20a}, portioning \citep{AAC+19a,EGL+24a}, and (fractional) participatory budgeting \citep{FGM16a,AzSh21a}.} More specifically, the central problem of budget division is how to distribute a divisible resource among the candidates in a fair and structured way based on the voters' preferences over these candidates. Notably, this resource does not have to be money as budget division, e.g., also captures time-sharing problems (where we need to allocate time to various activities in a fixed time frame) or probabilistic voting (where we need to assign the probability to win the election to the candidates).

Perhaps due to these versatile applications, budget division has recently attracted significant attention \citep[e.g.,][]{AAC+19a,ABM20a,BBPS21a,BBP+19a,FPPV21a,BGSS24a,EGL+24a,EKPS24a} and it has been studied based on various assumptions on the voters' preferences. For instance, \citet{AAC+19a} and \citet{EKPS24a} study budget division under the assumption that the voters rank all candidates, \citet{FGM16a} suppose that the voters have cardinal utilities, and \citet{FPPV21a} let voters report ideal resource distributions. By contrast, we will investigate budget division based on the assumption that the voters report approval ballots over the candidates, i.e., each voter only reports the set of candidates he approves of. This model has been first suggested by \citet{BMS05a} and it is considered in numerous recent works \citep[e.g.,][]{ABM20a,MPS20a,BBPS21a,BBP+19a} because approval preferences achieve a favorable trade-off between the cognitive burden on the voters and the expressiveness of the ballot format.\footnote{We refer to the work of \citet{BrFi07c} for a detailed discussion of the advantages of approval ballots and to the papers of \citet{FBG23a} and \citet{YHPP24a} for an experimental evaluation of approval ballots in the context of participatory budgeting.} The main study object of this paper thus are \emph{distribution rules}, which map every profile of approval ballots to an allocation of the resource to the candidates.

A central objective for many budget division problems is to choose outcomes that are fair towards all voters: an appropriate amount of the budget should be allocated to the approved candidates of every group of voters.
For instance, if each voter in our introductory example only approves of a single sport club, it seems desirable that every sport club gets a portion of the budget that is proportional to the number of voters approving it.
%The rationale behind this is that all voters should be able benefit from the allocation of the money to the sport clubs, thus necessitating a fair division of the money.
Moreover, a multitude of axioms has been suggested to formalize fairness also in the context of general  approval preferences \citep[see, e.g.,][]{BMS05a,Dudd15a,ABM20a,BBP+19a}, most of which are based on the idea that each voter controls a $\frac{1}{n}$ share of the total budget (where $n$ is the number of voters). Two of the strongest fairness axioms based on this approach are \emph{average fair share (AFS)} and \emph{core}, both of which were suggested by \citet{ABM20a}. Average fair share requires for every group of voters that approves a common candidate that the average utility of the voters in this group is at least as large as the utility they would obtain if the voters allocate their share to the commonly approved candidate. On the other hand, core applies the universal concept of core stability to budget division: there should not be a group of voters who can enforce an outcome that is better for each voter in the group by only using the share they deserve.
In the literature, there is only one distribution rule that is known to satisfy AFS and core---the Nash product rule (\nash)~\citep{ABM20a}.
This rule returns the distribution that maximizes the Nash social welfare, i.e., the product of the total share assigned to the approved candidates of each voter.
However, while being  fair, this rule is rather intricate and exhibits otherwise undesirable behavior.
In particular, the distribution returned by \nash may contain irrational values~\citep{AAC+19a}, which means that we can only approximate this rule in practice.
Closely connected to this is the issue that, for many profiles, it is difficult even for experts to verify whether an allocation is the one selected by \nash.\footnote{We note that, from a computational perspective, it is easy to verify the \nash distribution by checking the first-order optimality constraints. However, checking these constraints by hand is tedious, and it seems not helpful to convince laypersons of the correctness of the computed distribution.}
Furthermore, \nash fails \emph{monotonicity}~\citep{BBPS21a}: if a candidate gains more support in the sense that an additional voter approves it, its share may decrease.
For instance, in our introductory example, this may lead to paradoxical situations where it would be better for a sport club to be approved by less voters.
Finally, we will show  in this paper that \nash also behaves counter-intuitively with respect to \emph{population consistency conditions}. Such conditions are well-studied in social choice theory  \citep[e.g.,][]{Smit73a,Youn75a,Fish78d,YoLe78a,Bran13a,BrPe19a} and roughly require that, when combining disjoint elections into one election, the outcome for the combined election should be consistent with the outcome for the subelections. In more detail, we will prove that, for \nash, it can happen that the share of a candidate in a combined election is less than its share in each subelection, even if the considered candidate obtains the maximal share in each subelection and the distributions chosen for the subelections are structurally similar. For our sport club example, this means that the share of a sport club in a city-wide election may be less than the minimal share the sport club would obtain in district-level elections, even if the sport club obtains the largest share in each district and the distributions chosen for the districts are structurally similar.

%Lastly, we will show in this paper that \nash also fails \emph{top population consistency}.
%This condition formalizes that, if a candidate receives the largest share of the budget in two disjoint elections, then it should also receive the largest share in a combined election.
%Since \nash fails this axiom, it can, e.g., in our sport club example happen that a sport club would receive the maximal share of the budget in each district of a city, but not if the election is held on a city-wide level. We note that similar population consistency conditions feature in numerous prominent results in social choice theory~\citep{Smit73a,Youn75a,Fish78d,YoLe78a,Bran11c,BrPe19a}.

Motivated by these observations, our central research question is whether there are simple combinatorial rules that exhibit more consistent behavior than \nash while satisfying strong fairness guarantees. However, as it seems out of range to achieve full AFS or core with such rules, we will consider approximate variants of these axioms called $\alpha$-AFS and $\alpha$-core, which relax the constraints imposed by the original axioms by a $\frac{1}{\alpha}$ factor. Thus, $1$-AFS and $1$-core correspond to AFS and core, and the fairness guarantees of $\alpha$-AFS or $\alpha$-core become weaker as $\alpha$ increases. These approximate fairness notions allow us to quantify the fairness of every distribution rule by determining the minimal $\alpha$ for which it satisfies $\alpha$-AFS or $\alpha$-core, so we can derive a much more fine-grained picture about the fairness guarantees of distribution rules based on these axioms.
Our formal goal is thus to identify distribution rules that satisfy monotonicity and demanding population consistency conditions while satisfying $\alpha$-AFS or $\alpha$-core for as small $\alpha$ as possible.
%\mash{As the rules that we consider do not satisfy AFS or core, we focus on multiplicative approximations of these notions. More formally, an outcome is an $\alpha$-approximation to AFS or core if each voter receives at least a $\frac{1}{\alpha}$ fraction of the utility that they  would have received from an AFS or core outcome.   }\mashsays{maybe I can move this to a place that fits more naturally}

\paragraph{Contributions} 
% \mashsays{some things here needs to be rephrased in light of the restructuring of the paper}
In an attempt to find simple, consistent, and fair distribution rules, we first analyze several classical rules such as the conditional utilitarian rule \citep{Dudd15a} and the fair utilitarian rule \citep{BMS02a,BBPS21a}.
However, all of these rules give poor approximations to AFS and core and they typically violate even weak population consistency conditions (see \Cref{tab:alg_comparison} for details), thereby showing that none of these rules meets our criteria.

Motivated by these negative results, we suggest the \emph{maximum payment rule} (\map) for approval-based budget division.\footnote{We note that \map was considered before under the name ``majoritarian portioning'' by \citet{SpGe19a} and \citet{BGP+22a} in related settings. See the related work section for details.}
Under \map, we first distribute the budget uniformly to the voters, who spend it on the candidates in multiple rounds. In more detail, in every round, \map identifies the candidate $x$ that is approved by the largest number of voters who have not spent their budget yet and all corresponding voters send their complete budget to this candidate. The final shares of all candidates are given by the total budget the voters spent on them. As our main result, we will show that \map answers our question about the existence of simple, fair, and consistent distribution rules in the affirmative. Specifically, we will prove that \map is monotonic and gives strong fairness guarantees as it satisfies $2$-AFS and $\Theta(\log n)$-core. Moreover, \map satisfies a demanding population consistency condition which we call \emph{ranked population consistency} (RPC). 

The rough idea of RPC is that, when a candidate is treated similarly in two disjoint subelections, then the share of this candidate in the joint election should be upper and lower bounded by its maximal and minimal share in the subelections. More specifically, for this axiom, we consider the rankings obtained by ordering the candidates with respect to their shares in the subelections. If the prefixes of these rankings agree up to some candidate, then the share of this candidate in the joint election should be upper and lower bounded by its maximal and minimal share in the subelections. Or, less formally, if the distributions chosen for two disjoint subelections agree on the candidate with the highest share, RPC bounds the share of this candidate in the joint election; if the distributions also agree on the candidate with the second highest share, RPC applies again to this candidate; etc. We believe this condition to be desirable because it ensures that the distribution chosen for the joint election closely resembles the distributions chosen for its subelections if these distributions are structurally similar. %Moreover, we note that stronger population consistency notions turn out impossible to satisfy with desirable distribution rules. 

On the negative, it needs to be mentioned that \map fails efficiency. Specifically, we show that \map can choose a distribution for which it is possible to increase the utility of every voter by a $\Theta(\log n)$ factor. However, in practice, we believe that it is rare for \map to return (severely) inefficient outcomes for a number of reasons. Firstly, in real-world elections, the number of candidates $m$ is typically rather low, which allows to circumvent our negative result. In particular, the proof of Theorem~2 of \citet{MPS20a} entails that it is not possible for \map to increase the utility of every voter by more than $\frac{\sqrt{m}}{2-2/\sqrt{m}}$. Secondly, we note that our worst-case instance requires a rather intrinsic structure and, more generally, that it seems for typical elections with many voters and few candidates unlikely that there are many inefficient outcomes at all. %Lastly, we note that, in numerous closely related settings, the recent literature \citep[e.g.,][]{LaSk20a,FPPV21a,EKPS24a} follows our approach as it also sacrifices efficiency in favor of fairness.

\begin{table}
	\centering
	\begin{tabular}{@{}ccccc@{}}
		\toprule
		& Monotonicity & Population consistency & $\alpha$-AFS & $\alpha$-core \\
		\midrule
		\nash & \xmark & \cellcolor{black!20} weak & 1 & 1 \\
		\cut & \cmark & \cellcolor{black!20}\xmark & \cellcolor{black!20} $\Theta(n)$ & \cellcolor{black!20} $\Theta(n)$ \\
		\fut & \cmark & \cellcolor{black!20} \xmark & \cellcolor{black!20} $\Theta(n)$ & \cellcolor{black!20} $\Theta(n)$ \\
		\ues & \cmark & \cellcolor{black!20} strong & \cellcolor{black!20} $\Theta(n)$ & \cellcolor{black!20} $\Theta(n)$ \\
		\cellcolor{black!20} \map & \cellcolor{black!20} \cmark & \cellcolor{black!20} ranked & \cellcolor{black!20} $2$ & \cellcolor{black!20} $\Theta(\log n)$ \\
		\cellcolor{black!20} \mps{\frac{1}{3}} & \cellcolor{black!20} \xmark & \cellcolor{black!20} ranked & \cellcolor{black!20} $\frac{3}{2}$ & \cellcolor{black!20} $O(\log n)$ \\
		\bottomrule
	\end{tabular}
	\caption{Analysis of distribution rules with respect to monotonicity, population consistency, and their approximation ratios to AFS and core. Each row states for a distribution rule whether it satisfies monotonicity, the degree to which it satisfies population consistency, and its approximation ratio for AFS and core. The first four rows focus on established rules, while the fifth and sixth row discuss two rules designed in this paper. For population consistency, the symbol \xmark indicates that the rule fails all population consistency conditions considered in this paper, whereas ``weak'', ``ranked'', and ``strong'' indicate the exact population consistency condition that is satisfied (see \Cref{subsec:popCons} for details).
		Contributions of this paper are shaded gray.}
	\label{tab:alg_comparison}
\end{table}

%This latter axiom formalizes a demanding population consistency condition which roughly requires that, if a rule chooses distributions that are structurally similar for two disjoint elections, then the share of every candidate in the joint election should be lower and upper bounded by its minimal and maximal share in the disjoint subelections. More specifically, for this axiom, we consider two distributions as similar if the orders obtained by sorting the candidates according to their shares are the same for both profiles. We believe this to be reasonable because, intuitively, there is relatively little conflict between distributions that agree on which candidates deserve a larger share than others. 
%Rules that satisfy RPC guarantee that the outcome for a joint election is closely related to the distributions selected for the subelections when these distributions are not conflicting and they thus afford a high degree of explainability in such situations.

By generalizing the idea of \map, we further introduce a new family of distribution rules called \emph{sequential payment rules}.
Just as for \map, sequential payment rules distribute the budget uniformly to the voters, who will spend their shares on the candidates. However, in contrast to \map, sequential payment rules do not require the voters to spend their whole budget in one shot. Instead, these rules are defined by payment willingness functions, which specify how much of his remaining budget a voter is willing to spend on his approved candidates in the next round of the rule. A sequential payment rule then works iteratively and, in every round, it identifies the candidate that maximizes the total payment willingness, let the voters send their corresponding payments to this candidate, and remove it from consideration. Finally, the distribution chosen by a sequential payment rule is proportional to the budgets accumulated by the candidates.

Our findings for this class of rules further strengthen our case for \map. In more detail, while all sequential payment rules satisfy RPC, we show that all these rules except for \map and one other rule called \ues fail monotonicity. Since \ues gives a poor approximation to AFS, this demonstrates that \map is the only sequential payment rule that satisfies our original desiderata. Moreover, we analyze the approximation ratio of sequential payment rules to AFS and show that no such rule satisfies $\alpha$-AFS for $\alpha<\frac{3}{2}$. This means that these rules only allow for a small improvement for the approximation ratio to AFS compared to \map. Lastly, we prove that our lower bound on the approximation ratio of sequential payment rules to AFS is tight. Specifically, we show that the \emph{$\frac{1}{3}$-multiplicative sequential payment rule (\mps{\frac{1}{3}})}, where the payment willingness of a voter is discounted by a factor of $\frac{1}{3}$ whenever he pays for a candidate, satisfies $\frac{3}{2}$-AFS. Moreover, based on a more fine-grained analysis that takes the maximum ballot size into account, we characterize \mps{\frac{1}{3}} as the fairest sequential payment rule. %Hence, this rule may be compelling for applications where simplicity and fairness considerations outweigh efficiency and monotonicity violations.

\paragraph{Related Work}

The problem of budget division with dichotomous preferences was first considered by \citet{BMS05a}, who, inspired by randomized social choice, defined several distribution rules and studied them with respect to strategyproofness, efficiency, and fairness. In particular, \citet{BMS05a} realized that there is a trade-off between these conditions because none of their rules simultaneously satisfies all three properties. This trade-off was further analyzed by \citet{Dudd15a} and \citet{BBPS21a}, the latter of whom proved that no strategyproof~and efficient rule can satisfy minimal fairness requirements. In a similar vein, \citet{MPS20a} showed that fair distribution rules only give poor guarantees on the utilitarian social welfare. 

On a more positive note, \citet{ABM20a} investigated distribution rules with respect to fairness notions and, in particular, showed that the Nash product rule satisfies AFS and core. Motivated by this result, several other authors have analyzed further aspects of \nash \citep[e.g.,][]{GuNe14a,AAC+19a,BBP+19a}. In particular, \citet{GuNe14a} studied approval-based distributions rules that return convex sets of distributions and characterized the rule that returns the set of all \nash distributions as the only rule that satisfies a mild proportionality condition, continuity, efficiency, independence of candidates who obtain a share of~$0$, and population consistency. Their population consistency axiom requires that, if the chosen sets of distributions intersect for two disjoint profiles, exactly the distributions in the intersection of these sets should be selected for the combined profile. This condition can be seen as the counterpart of what we call weak population consistency for the set-valued setting.

We further note that \citet{BBP+19a} presented a dynamic process for computing the Nash product rule, which was first suggested by \citet{Cove84a} in the context of portfolio selection in stock markets. More precisely, these authors demonstrated that, when the voters repeatedly update their spending such that it is proportional to the restriction of the current distribution to their approved alternatives, the distribution will converge against the one chosen by \nash. Similar redistribution dynamics have recently attracted significant attention \citep[e.g.,][]{FPPV21a,BBP+19a,BGSS23b,HKMS24a}, but these dynamics significantly differ from our work as we do not allow voters to reallocate spent money.

Furthermore, our paper belongs to an extensive stream of research that studies fairness in budget division. In more detail, besides the paper by \citet{ABM20a}, fairness notions similar to AFS and core have been studied in budget division models with strict preferences~\citep{AAC+19a,EKPS24a}, cardinal utilities~\citep{FGM16a}, and ideal distributions~\citep{FPPV21a,CCP24a,BGSS24a,EGL+24a,FrSc24a}, as well as budget division models with additional structure on the outcomes \citep{BLS24a,LPA+24a}.
%In more detail, the paper by \citet{ABM20a} sparked a large number of follow-up works: fairness notions similar to AFS and core have been studied in budget division models with strict preferences \citep{AAC+19a,EKPS24a}, cardinal utilities \citep{FGM16a}, and ideal distributions \citep{FPPV21a,CCP24a,BGSS24a,EGL+24a}.
For instance, \citet{FPPV21a} designed a strategyproof and fair distribution rule under the assumption that the voters report ideal budget division.\footnote{We note that the same setting was also considered in earlier works \citep[e.g.,][]{Intr73a,LNP08a}, but these papers do not study on fairness conditions.}
Moreover, analogous fairness concerns have attracted significant attention in settings related to budget division, such as committee voting \citep{FSST17a,LaSk22b} and participatory budgeting \citep{AzSh21a,RSM25a}.

In such related setting, variants of the maximum payment rule have been considered before \citep{SpGe19a,BGP+22a}. In more detail, \citet{BGP+22a} referred to \map as ``majoritarian portioning'' and use it as an auxiliary tool: they apply \map to turn an approval profile into a vote count over the parties and then apply apportionment methods to this vote count to select a fixed-size multiset of the parties. However, their analysis only concerns such combined portioning-apportionment methods. Similarly, \citet{SpGe19a} suggested \map to turn an approval profile into an assignment of votes to parties, with the ultimate goal again being to elect a discrete committee based on this vote count. We further note that these authors also make some informal claims about the properties of \map.

Lastly, we note that budget division is closely related to randomized social choice \citep{Bran17a}, where the goal is to select a probability distribution over the candidates from which the final election winner will be chosen by chance. Specifically, randomized social choice can be seen as a particular application of budget division, where the shares of candidates are interpreted as their probabilities of being selected as the final winner. We, however, note that randomized social choice has an \emph{ex post} perspective, which significantly affects the considered axioms. For instance, in randomized social choice, it seems desirable to minimize the amount of randomization, so that the final winner depends as much as possible on the voters' preferences instead of chance. Such considerations tend to conflict with the fairness notions studied in budget division, which aim to spread the budget among the candidates to satisfy all voters. 
%While the model of randomized social choice is mathematically identical to that of budget division, the goals, and thus also the considered axioms, are different. In particular, in randomized social choice, using a large amount of randomization is often met with criticism as this places a large degree of uncertainty on the final winner, whereas allocating the budget to many candidates seems uncontroversial in budget division. %Next, in participatory budgeting and committee voting, the goal is to select (indivisible) projects or candidates subject to budget constraints and a target committee size, respectively. Hence, these problems can be seen as more constrained, discrete counterparts of budget division where it is not possible to arbitrarily split the budget between candidates. Because of the discrete nature of these problems, different fairness notions are studied for these models.

%%%%%%%%%%%%%%%%%%%%%%%%%%%%%%%%%%%%%%%%%%%%%%%%%%%%%%%%%%%%%%%%%%%%%%%%
\section{Model and Axioms}
\label{sec:prelim}

For any positive integer~$t \in \mathbb{N}$, we define $[t]= \{1, 2, \dots, t\}$.
Let $N=[n]$ be a set of~$n$ voters and $C=\{x_1,\dots, x_m\}$ a set of~$m$ candidates.
We assume that the voters have \emph{dichotomous} preferences over the candidates, i.e., voters only distinguish between approved and disapproved candidates. Moreover, each voter $i\in N$ reports his preferences in the form of an  \emph{approval ballot} $\approvalBallotOfvoter{i}$, which is the non-empty set of his approved candidates. %The set of all approval ballots is thus $2^C\setminus \{\emptyset\}$.
An \emph{approval profile} $\mathcal{A} = (\approvalBallotOfvoter{1}, \approvalBallotOfvoter{2}, \dots, \approvalBallotOfvoter{n})$ contains the approval ballots of all voters. Given an approval profile $\mathcal A$, we denote by $N_x(\mathcal A)= \{i \in N \colon x \in \approvalBallotOfvoter{i}\}$ the set of voters who approve of candidate~$x$ and we call $|N_x(\mathcal A)|$ the \emph{approval score} of candidate~$x$. %In this paper, we will allow for a variable sets of voters and candidates, and we thus denote by $N_\mathcal{A}$ and $C_\mathcal{A}$ the sets of voters and candidates in $\mathcal{A}$. More formally, $N_\mathcal{A}$ is the domain of $\mathcal{A}$ and $2^{C_\mathcal{A}}\setminus \{\emptyset\}$ its codomain.
Finally, an \emph{election instance} is a tuple $\mathcal{I} = (N,C,\mathcal{A})$ that specifies the set of voters $N$, the set of candidates $C$, and an approval profile $\mathcal A$ for the given sets of voters and candidates. Following the literature, we will typically equate election instances with the corresponding approval profiles as these contain all relevant information. Because we will study population consistency axioms in this paper, we define by $\mathcal{A}+\mathcal{A}'$ the approval profile that combines two voter-disjoint profiles $\mathcal{A}$ and $\mathcal{A'}$. More formally, given two instances $\mathcal{I}=(N,C,\mathcal{A})$ and $\mathcal{I}'=(N',C',\mathcal{A}')$ such that $C=C'$ and $N\cap N'=\emptyset$,
$\mathcal{A}+\mathcal{A}'$ is the approval profile corresponding to the instance $\mathcal{I}+\mathcal{I'}= (N\cup N', C, (\mathcal{A}, \mathcal{A}'))$.

Given an approval profile, our goal is to distribute a divisible resource to the candidates. To this end, we will use distribution rules which determine for every election instance an allocation of the resource to the candidates. We will formalize such resource allocations by distributions that specify the share of the resource assigned to each candidate. More formally, a \emph{distribution} $p$ is a function of the type $C\rightarrow [0,1]$ such that $\sum_{x\in C} p(x)=1$, and we denote by $\Delta(C)$ the set of all distributions over $C$. Then, a \emph{distribution rule} $f$  is a function that maps each approval profile $\mathcal{A}$ to a distribution $p \in {\Delta}(C)$. For the ease of notation, we define by $f(\mathcal{A}, x)$ the share of the resource that the distribution rule $f$ assigns to candidate $x$ under the approval profile $\mathcal{A}$. For instance, $f(\mathcal A, x)=0.5$ means that, for the profile $\mathcal{A}$, $f$ assigns half of the resource to candidate $x$.
Following the literature, we assume that the utility of a voter $i$ for a distribution $p$ is the total share assigned to his approved candidates, i.e., the utility function of voter $i$ is given by $u_i(p)=\sum_{x\in A_i} p(x)$.

In the following three subsections, we will introduce the central axioms of this paper.

\subsection{Monotonicity}

As our first axiom, we will discuss monotonicity, which requires that the share of a candidate does not decrease if it gets additional supporters. Hence, this axiom can be seen as an incentive for candidates to gather as much support as possible.

% define candidate monotonicity
\begin{definition}[Monotonicity]
    A distribution rule $f$ is \emph{monotonic} if $f(\mathcal{A}', x)\geq f(\mathcal{A}, x)$ for all approval profiles $\mathcal{A}$ and $\mathcal{A'}$, voters $i\in N$, and candidates $x\in C$ such that $A'_i=A_i\cup \{x\}$ and $A'_j=A_j$ for all $j\in N\setminus \{i\}$.
\end{definition}

Analogous monotonicity notions are well-studied in numerous voting settings \citep[e.g.,][]{Gibb77a,Fish82a,Mask99a,SaZw10a,SaFi19a}, thereby emphasizing the significance of this axiom. 
%By contrast, in budget division with dichotomous preferences, monotonicity is known to be a rather restrictive concept as only few known rules satisfy this condition \citep{BMS05a,BBPS21a}. We nevertheless believe that this condition is a natural desideratum for budget divisions as it ensures that candidates are rewarded for obtaining more votes.

% define order population consistency
\subsection{Population Consistency}\label{subsec:popCons}

We will next discuss several population consistency axioms, which formalize that the chosen distribution for a joint election should be consistent with the distributions that are selected for its subelections. For instance, in our introductory sport club example, population consistency ensures that the distribution chosen for a city-wide election is, when possible, consistent with the distributions that would be selected for district-level elections. We believe such consistency notions to be important for at least two reasons. 
Firstly, population consistency adds justification for running elections on the city level instead of the district level. In more detail, for approval-based budget division, it would be easily possible to split the budget among the districts and let them run separate elections to allocate their shares. Consequently, we need to avoid situations where the outcome for the city-wide election directly conflicts with the outcomes of the district-level election as districts may not be willing to participate in the mechanism otherwise. Secondly, we note that population consistent distributions rules tend to afford higher explainability of the outcome as they allow to justify the selected distribution in terms of the distributions chosen for subelections. For instance, for population consistent rules, we may explain why a candidate did obtain a certain share by referring back to the outcomes for the individual districts.

Motivated by these observations, we will next present three population consistency axioms. We start by introducing the standard notion of population consistency, which states that, if a rule chooses the same outcome for two disjoint elections, it should choose the same outcome also for the combined election.
We call this condition weak population consistency as it is our weakest population consistency axiom, and we define it as follows.
\begin{definition}[Weak population consistency (WPC)]
	A distribution rule $f$ satisfies \emph{weak population consistency (WPC)} if $f(\mathcal A+\mathcal A')=f(\mathcal A)$ for all voter-disjoint profiles $\mathcal A$ and $\mathcal A'$ with $f(\mathcal A)=f(\mathcal A')$.
\end{definition}
Variants of this condition feature in numerous prominent results of social choice theory \citep[e.g.,][]{Smit73a,Youn75a,Fish78d,YoLe78a,Bran11c,LaSk21a,BrPe19a}. However, we believe that WPC is too weak in the context of budget division because the condition of choosing the exact same distribution in two profiles is too restrictive. For instance, assume a distribution rule $f$ chooses for two profiles $\mathcal A$ and $\mathcal A'$ the distributions $p=f(\mathcal A)$ and $q=f(\mathcal A')$ with $p(x)=1$ and $q(x)=0.99$. It then seems reasonable that candidate $x$ also gets a share close to $1$ in the combined election $\mathcal A+\mathcal A'$, but WPC does not impose any constraint on the distribution $f(\mathcal A+\mathcal A')$.

One possible solution for this is to demand that, for each alternative $x$, the share assigned to $x$ in the joint profile $\mathcal{A}+\mathcal {A}'$ should be lower and upper bounded by its shares in $\mathcal A$ and $\mathcal A'$. We formalize this idea next.
\begin{definition}[Strong population consistency (SPC)]
	A distribution rule $f$ satisfies \emph{strong population consistency (SPC)} if $\min\{f(\mathcal A,x), f(\mathcal A',x)\}\leq f(\mathcal A+\mathcal A',x)\leq \max\{f(\mathcal A,x), f(\mathcal A',x)\}$
	for all candidates $x\in C$ and voter-disjoint profiles $\mathcal A$ and $\mathcal A'$.
\end{definition}

We believe that SPC aligns well with our motivation for population consistency as it ensures that the outcome for an election is closely tied to the outcomes selected for its subelections. Specifically, it immediately allows to give explanations for the outcome in terms of the distributions selected for subelections, and it ensures that the outcome in a city-wide election is not too different from the distributions chosen for the districts if these distributions are not inherently at conflict. 
However, as we show next, SPC turns out to be overly restrictive for approval-based budget division as it conflicts even with minimal efficiency considerations.
To formalize this claim, we say a distribution rule $f$ is \emph{unanimous} if $f(\mathcal{A},x)=1$ for all profiles $\mathcal{A}$ and candidates $x$ such that $x$ is the unique candidate that is approved by all voters in $\mathcal{A}$. We next prove that unanimity and SPC are incompatible with each other.%\footnote{Strong population consistency can be satisfied by desirable rules in related settings. For instance, when voters report strict preferences, it can be shown that proportional scoring rules \citep[see, e.g.,][]{Barb79b,Bran17a} such as the harmonic rule by \citet{BCH+15a} or the random dictatorship by \citet{Gibb77a} satisfy SPC.} 

%this axiom turns out to be too prohibitive for dichotomous preferences. To see this, consider the profiles $\mathcal A$ and $\mathcal A'$ on the set of candidate $C=\{a,b,c\}$ such that all voters in $\mathcal A$ report $\{a,b\}$ and all voters in $\mathcal A'$ report $\{a,c\}$. For these profiles, it seems natural to choose the distributions $p$ and $q$ with $p(a)=p(b)=\frac{1}{2}$ and $q(a)=q(c)=\frac{1}{2}$, respectively. By contrast, in the profile $\mathcal A+\mathcal A'$, most rules will assign the entire budget to $a$ since it is the only candidate that is approved by all voters, but this violates SPC.\footnote{When slightly modifying our setting, strong population consistency can be satisfied by reasonable rules. For instance, when voters report strict preferences, it can be shown that proportional scoring rules \citep[see, e.g.,][]{Barb79b,Bran17a} such as the harmonic rule by \citet{BCH+15a} or the random dictatorship by \citet{Gibb77a} satisfy strong population consistency.}

\begin{proposition}\label{prop:SPC}
	No unanimous distribution rule satisfies SPC.
\end{proposition}
\begin{proof}
	Assume for contradiction that there is a unanimous distribution rule $f$ that satisfies SPC. We will derive a contradiction by considering the following profile $\mathcal{A}$ with four voters and four candidates. We note that we can extend the profile to larger numbers of candidates by adding candidates that are not approved by any voter. Moreover, our construction can be generalized to larger numbers of voters by adding voters that report $\{a_1,b_1\}$ (or any other ballot in $\mathcal{A}$) as adding such voters does not affect the argument. 
	\begin{center}
	\begin{tabular}{lllll}
		$\mathcal{A}$: & $1$: $\{a_1, b_1\}$ & $1$: $\{a_1, b_2\}$ & $1$: $\{a_2, b_1\}$ & $1$: $\{a_2, b_2\}$
	\end{tabular}
	\end{center}
	
	Firstly, we note that $\mathcal{A}$ can be decomposed into two profiles $\mathcal{A}^1$ and $\mathcal{A}^2$, where $a_1$ and $a_2$ are unanimously approved. 
		\begin{center}
		\begin{tabular}{llllll}
			$\mathcal{A}^1$: & $1$: $\{a_1, b_1\}$ & $1$: $\{a_1, b_2\}$\hspace{3cm} & $\mathcal{A}^2$: & $1$: $\{a_2, b_1\}$ & $1$: $\{a_2, b_2\}$
		\end{tabular}
	\end{center}
	Unanimity requires for these profiles that $a_1$ and $a_2$ get a share of $1$, respectively. In turn, this means that $b_1$ and $b_2$ get a share of $0$ in both profiles, so SPC requires that $f(\mathcal A, b_1)=f(\mathcal A, b_2)=0$.

	Secondly, we can also decompose $\mathcal{A}$ into two profiles $\mathcal{A}^3$ and $\mathcal{A}^4$, where $b_1$ and $b_2$ are unanimously approved, respectively.
	\begin{center}
		\begin{tabular}{llllll}
			$\mathcal{A}^3$: & $1$: $\{a_1, b_1\}$ & $1$: $\{a_2, b_1\}$\hspace{3cm} & $\mathcal{A}^4$: & $1$: $\{a_1, b_2\}$ & $1$: $\{a_2, b_2\}$
		\end{tabular}
	\end{center}
	 For these profiles, unanimity postulates that $a_1$ and $a_2$ both obtain a share of $0$. In turn, SPC implies that $f(\mathcal A, a_1)=f(\mathcal A, a_2)=0$. However, this means that every alternative obtains a share of $0$ in $\mathcal{A}$, which contradicts the definition of a distribution. Hence, our assumption is wrong and no unanimous distribution rule satisfies SPC.
\end{proof}

Intuitively, the proof of \Cref{prop:SPC} relies on the fact that SPC entails a prohibitively high degree of independence between the shares of candidates. In particular, if a candidate $x$ receives a share of $0$ in two disjoint elections, SPC requires that it also receives a share of $0$ in the joint election, even if the shares of the other candidates differ vastly. To circumvent \Cref{prop:SPC}, we will thus introduce an intermediate population consistency notion that allows for more dependencies between the candidates while still giving strong consistency guarantees. Specifically, the idea of \emph{ranked population consistency (RPC)} is to consider the ordinal rankings obtained by ordering the candidates in decreasing order of their assigned shares and to require SPC only for candidates $x$ such that the prefixes of these rankings up to candidate $x$ agree. To make this idea more formal, we define the \emph{distribution ranking} $\succsim^p$ of a distribution $p$ as the binary relation over $C$ given by $x\succsim^p y$ if and only if $p(x)\geq p(y)$ for all $x,y\in C$. Further, we define by ${\succsim^p}|_{x}$ the restriction of $\succsim^p$ to candidates that have a share of at least $p(x)$, i.e., ${\succsim^p}|_{x}$ is the restriction of $\succsim^p$ to the set $\{y\in C\colon p(y)\geq p(x)\}$. Then, ranked population consistency requires that strong population consistency only holds for profiles $\mathcal{A}$ and $\mathcal {A'}$ and candidates $x\in C$ such that ${\succsim^{f(\mathcal A)}}|_{x} ={ \succsim^{f(\mathcal A)}}|_{x}$.

\begin{definition}[Ranked population consistency (RPC)]
	A distribution rule $f$ satisfies \emph{ranked population consistency (RPC)} if
	$\min\{f(\mathcal A,x), f(\mathcal A',x)\}\leq f(\mathcal A+\mathcal A',x)\leq \max\{f(\mathcal A,x), f(\mathcal A',x)\}$
	for all candidates $x\in C$ and voter-disjoint profiles $\mathcal A$ and $\mathcal A'$ such that ${\succsim^{f(\mathcal A)}}|_{x}={\succsim^{f(\mathcal A')}}|_{x}$.
\end{definition}

Less formally, we may imagine for RPC that we go through the distribution rankings from top to bottom and apply the SPC condition as long as the rankings agree: if $f(\mathcal{A})$ and $f(\mathcal{A}')$ assign the maximal share to the same candidate, SPC needs to hold for this candidate in $\mathcal A + \mathcal A'$. If these distributions further agree on the candidate with the second highest share, SPC needs to hold for this candidate, too. This reasoning is repeated until we arrive at the first disagreement in the distribution rankings of $f(\mathcal A)$ and $f(\mathcal A')$. From this description, it should be clear that SPC implies RPC. Moreover, RPC implies WPC because the distribution rankings for $f(\mathcal{A})$ and $f(\mathcal{A}')$ agree when $f(\mathcal{A})=f(\mathcal{A}')$. We further observe that the proof of \Cref{prop:SPC} does not work with RPC and, as we will soon see, there is indeed an appealing distribution rule that satisfies unanimity and RPC. We therefore believe that RPC is a desirable condition as it preserves the normative appeal of SPC for ``structurally similar'' distributions while not being overly restrictive. Lastly, it may seem even more intuitive to require for the definition of RPC that ${\succsim^{f(\mathcal A+\mathcal A')}}|_x={\succsim^{f(\mathcal A)}}|_x$ when ${\succsim^{f(\mathcal A)}}|_x={\succsim^{f(\mathcal A')}}|_x$. We decided against this notion since it is unrelated to WPC and SPC, but we observe that all of our results remain true for this alternative definition.

\subsection{Fairness}

A central concern in budget division problems is fairness: we should select outcomes that guarantee an appropriate amount of utility to all voters. In budget division, fairness is typically formalized based on the idea that each voter deserves to control a share of $\frac{1}{n}$ of the total budget. Maybe the most direct formalization of this idea is \emph{decomposability} \citep{BBP+19a}, which requires that the selected distribution $p$ can be decomposed into payments from individual voters such that each voter pays exactly $\frac{1}{n}$ to his approved candidates.

\begin{definition}[Decomposability]
  A distribution $p$ is \emph{decomposable} for an approval profile $\mathcal{A}$ if there exist individual distributions $\{p^i\}_{i\in N}$ such that (i) $p = \frac{1}{n}\sum_{i\in N} p^i$ and (ii) $p^i(x) = 0$ for all $i\in N, x\notin A_i.$ A distribution rule $f$ is decomposable if $f(\mathcal{A})$ is decomposable for every profile $\mathcal{A}$.
\end{definition}

While decomposability formalizes a basic fairness notion, it fails to capture group synergies between voters. We thus view this condition as a minimal fairness requirement and only discuss decomposable rules in this work. By contrast, we will measure the fairness of distribution rules by studying stronger axioms in a quantitative manner. In particular, we will investigate the approximation ratio of distribution rules to average fair share (AFS) and core, two of the strongest fairness conditions in budget division. We will next discuss these two axioms, which were both suggested by \citet{ABM20a} in the context of budget division.\footnote{The core was first suggested by \citet{Auma61b} in the context of cooperative games and it has recently been used as fairness notion in numerous social choice settings \citep[e.g.,][]{FGM16a,ABC+16a,EKPS24a}.}
%\jvsays{Should we mention an older citation for the inspiration of core?}
\begin{itemize}
	\item \emph{Average fair share (AFS)} lower bounds the average utility of every group of voters $S$ that support a common candidate $x$ by the average utility this group would enjoy if all voters in $S$ allocate their $\frac{1}{n}$ share to $x$. More formally, a distribution $p$ satisfies AFS for a profile $\mathcal{A}$ if $ \frac{1}{|S|} \sum_{i \in S} u_i(p) \geq \frac{|S|}{n}$ for all groups of voters $S\subseteq N$ such that $\bigcap_{i\in S} A_i\neq \emptyset$.
	\item \emph{Core} requires that no group of voters $S$ can find a distribution that only uses the $\frac{1}{n}$ shares of the voters in $S$, weakly increases the utility of all voter in $S$, and strictly increases the utility of some voter in $S$. More formally, a distribution $p$ satisfies core for a profile $\mathcal A$ if there is no distribution $q$ and group of voters $S$ such that $\frac{|S|}{n} u_i(q)\geq u_i(p)$ for all $i\in S$ and $\frac{|S|}{n} u_i(q)> u_i(p)$ for some $i\in S$.
\end{itemize}

We note that both of these axioms are very challenging as they give strong lower bounds on the utilities of the voters. To allow for a more fine-grained analysis, we will thus consider approximate versions of these axioms, called $\alpha$-average fair share ($\alpha$-AFS) and $\alpha$-core, which relax the corresponding axioms by introducing a multiplicative approximation factor. In more detail, $\alpha$-AFS weakens the lower bounds on the average utilities of groups of voters with a commonly approved candidate, whereas $\alpha$-core rules out that there is a set of voters $S$ such that the voters in $S$ can obtain $\alpha$ times their original utility by only using their shares.
%\jvsays{This is confusing as it's not how we've re-written the definition. They are mathematically equivalent but seems confusing for the unknowing reader.}

\begin{definition}[$\alpha$-AFS]
A distribution~$p$ satisfies \emph{$\alpha$-AFS} for an approval profile $\mathcal{A}$ and an approximate factor $\alpha\geq 1$ if
$\alpha\cdot \frac{1}{|S|} \cdot \sum_{i \in S} u_i(p) \geq \frac{|S|}{n}$
 for all~$S \subseteq N$ with $\bigcap_{i \in S} A_i \neq \emptyset$.
\end{definition}

\begin{definition}[$\alpha$-core]
	A distribution~$p$ satisfies \emph{$\alpha$-core} for an approval profile $\mathcal{A}$ and an approximation factor $\alpha\geq 1$ if for all~$S \subseteq N$, there is no distribution~$q$ such that $\frac{|S|}{n} \cdot u_i(q)\geq \alpha \cdot u_i(p) $ for all~$i \in S$, and there is some voter~$i \in S$ such that $ \frac{|S|}{n} \cdot u_i(q)>\alpha \cdot u_i(p)$.
\end{definition}

Moreover, we say that a distribution rule $f$ satisfies $\alpha$-AFS (resp., $\alpha$-core) for an approximation factor $\alpha\geq 1$ if $f(\mathcal A)$ satisfies $\alpha$-AFS (resp., $\alpha$-core) for every approval profile $\mathcal{A}$. These notions allow us to quantify the fairness of distribution rules by analyzing the minimal value of $\alpha$ for which they satisfy the given axioms. In particular, a smaller value of $\alpha$ corresponds to a higher degree of fairness and $1$-AFS (resp., $1$-core) is equivalent to the original notion of AFS (resp., core). We note that similar multiplicative approximations of fairness axioms have been considered before, albeit not in the context of budget division with dichotomous preferences \citep[e.g.,][]{PeSk20a,MSWW22a,EKPS24a,BLS24a}.

\begin{example}[Fairness axioms]
	We will next discuss an example to further illustrate our fairness axioms. To this end, consider the following approval profile $\mathcal A$ with $6$ voters and $4$ candidates.\medskip
	
{
	\centering	
	\begin{tabular}{@{}lllllll@{}}
		$\mathcal A$: & $1$: $\{a,b_1\}$ & $1$: $\{a,b_2\}$ & $1$: $\{a,b_3\}$ & $1$: $\{b_1\}$ & $1$: $\{b_2\}$ & $1$: $\{b_3\}$
	\end{tabular}\par
}\smallskip

	Moreover, let $p$ denote the distribution given by $p(b_1)=p(b_2)=p(b_3)=\frac{1}{3}$. First, we note that this distribution is decomposable for $\mathcal A$ as demonstrated by the following decomposition: the first and fourth voter pay their $\frac{1}{6}$ share to $b_1$, the second and fifth voter pay for $b_2$, and the third and sixth voter pay for $b_3$. However, the distribution fails both AFS and core because, intuitively, the voters approving $a$ deserve half of the budget but each of them only has a utility of $\frac{1}{3}$. We thus compute the approximation ratio of $p$ to core and AFS, and define to this end $S$ as the group of voters approving $a$. By its definition, $\alpha$-AFS requires that $\frac{\alpha}{3}=\alpha \frac{1}{|S|} \sum_{i\in S} u_i(p)\geq \frac{|S|}{|N|}=\frac{1}{2}$. Hence, $p$ can only satisfy $\alpha$-AFS for $\alpha\geq \frac{3}{2}$ and it turns out that it actually satisfies $\frac{3}{2}$-AFS by checking the constraints for all other candidates. On the other hand, for $\alpha$-core, consider the distribution $q$ given by $q(a)=1$. For this distribution, it holds for all $i\in S$ that $\frac{1}{2}=\frac{|S|}{|N|} u_i(q)\geq u_i(p)=\frac{1}{3}$, so it follows that $p$ can only satisfy $\alpha$-core for $\alpha\geq \frac{3}{2}$. Moreover, it can again be checked that $p$ satisfies $\frac{3}{2}$-core by considering all other groups of voters. Hence, while the distribution $p$ may not fully satisfy AFS and core, it is still rather fair.
	\end{example}

As the last point in this section, we will discuss the relation between $\alpha$-AFS and $\alpha$-core. To this end, we first recall that AFS and core are logically incomparable, i.e., there are instances where a distribution satisfies AFS but not core and vice versa \citep{ABM20a}. By contrast, we will next prove that the approximation ratios to core and AFS are mathematically related. We note that the implications given in the following proposition are asymptotically tight (see Remark \ref{para:tight} for details).

\begin{restatable}{proposition}{afscore}\label{prop:afs-core}
	If a distribution $p$ satisfies $\alpha$-AFS for an approval profile $\mathcal{A}$, it satisfies $\alpha(1+\log(n))$-core. Conversely, if a distribution $p$ satisfies $\alpha$-core for an approval profile $\mathcal{A}$, it satisfies $2\alpha$-AFS.
\end{restatable}

\begin{proof}
	We will first show the implication from approximate AFS to approximate core and thus let $p$ denote a distribution satisfying $\alpha$-AFS for an approval profile $\mathcal A$. To show that $p$ satisfies $\alpha(1+\log n)$-core, we need an auxiliary concept called \emph{proportional fairness (PF)}. We thus say that a distribution $p$ satisfies $\beta$-PF for an approval profile $\mathcal{A}$ if
	$$
	\text{PF}(p) \coloneqq \max \limits_{q\in \Delta(C)} \frac{1}{n}\sum\limits_{i\in N} \frac{u_i(q)}{u_i(p)} \leq \beta.
	$$
	\citet{EKPS24a} have shown that if a distribution satisfies $\beta$-PF for a profile $\mathcal A$, then the distribution also satisfies $\beta$-core. Thus, it suffices to show that $p$ satisfies $\alpha(1+\log(n))$-PF. To this end, we note analogous to~\citep{EKPS24a} that
	\begin{align}
	\label{eq:PF}
		\text{PF}(p) &= \max \limits_{q\in \Delta(C)} \frac{1}{n}\sum\limits_{i\in N} \frac{u_i(q)}{u_i(p)}
		= \max \limits_{x\in C} \frac{1}{n} \sum\limits_{i\in N} \frac{u_i(x)}{u_i(p)}
		=  \max \limits_{x\in C} \frac{1}{n} \sum\limits_{i\in N_x(\mathcal A)} \frac{1}{u_i(p)}.
	\end{align}
	Consider now an arbitrary candidate $x$. Since $p$ satisfies $\alpha$-AFS, it holds that
		$\frac{1}{|N_x(\mathcal A)|} \sum_{i\in N_{x}(\mathcal A)} u_i(p) \geq \frac{1}{\alpha} \cdot\frac{|N_x(\mathcal A)|}{n}$.
	Since the average utility of voters in $N_x(\mathcal A)$ is at least $\frac{1}{\alpha}\cdot \frac{|N_x(\mathcal A)|}{n}$, there must exist a voter $i_1\in N_x(\mathcal A)$ who gets utility $u_{i_1}(p)\geq \frac{1}{\alpha} \cdot\frac{|N_x(\mathcal A)|}{n}$. Consider now the group of voters $N_x(\mathcal A)\setminus \{i_1\}$. By applying the $\alpha$-AFS bound to this group, we derive that  the average utility of voters in $N_x(\mathcal A)\setminus \{i_1\}$ is at least $\frac{1}{\alpha} \cdot \frac{|N_x(\mathcal A)|-1}{n}$. Hence, there is a voter $i_2$ such that $u_{i_2}(p)\geq \frac{1}{\alpha}\cdot \frac{|N_x(\mathcal A)|-1}{n}$.  By applying this argument iteratively, we see that there exists an order of the voters in $N_x(\mathcal A) =\{i_1, \dots, i_{|N_x(\mathcal A)|}\}$ such that $u_{i_t}(p) \geq \frac{1}{\alpha}\cdot \frac{|N_x(\mathcal A)|-(t-1)}{n}$ for each $t\in[|N_x(\mathcal A)|]$. Using this bound in conjunction with \Cref{eq:PF}, we infer that
	\begin{align*}
		\text{PF}(p)  &=  \max \limits_{x\in C} \frac{1}{n} \sum\limits_{i\in N_x(\mathcal A)} \frac{1}{u_i(p)}\\
		 %=  \max \limits_{x\in C} \frac{1}{n} \sum_{t=1}^{|N_x(\mathcal A)|} \frac{1}{u_{i_t}(p)} 
		&\leq  \max \limits_{x\in C} \frac{1}{n} \sum_{t=1}^{|N_x(\mathcal A)|} \frac{\alpha n}{ |N_x(\mathcal A)|-(t-1)}\\
		& \leq \alpha \cdot \max_{x \in C} \left( 1 + \log(|N_x(\mathcal{A})|) \right) \\
		&\leq \alpha (1 + \log(n)).
	\end{align*}
Hence, $p$ satisfies $\alpha (1+\log(n))$-PF, and therefore also $\alpha(1+\log(n))$-core.

We now show that if $p$ satisfies $\alpha$-core for an approval profile $\mathcal A$, then it also satisfies $2\alpha$-AFS. For this, assume that $p$ is a distribution that satisfies $\alpha$-core for an approval profile $\mathcal A$. Moreover, we consider an arbitrary group of voters $S$ with $\bigcap_{i\in S} A_i \neq \emptyset$, let $x\in \bigcap_{i\in S} A_i$, and let $q$ denote the distribution with $q(x)=1$.
Because $p$ satisfies $\alpha$-core and $x\in\bigcap_{i\in S} A_i$, there exists a voter $i_1\in S$ such that $\alpha \cdot u_{i_1}(p)\geq {\frac{|S|}{n} \cdot u_{i_1}(q)}=\frac{|S|}{n}$. Next, by analyzing the group $S\setminus \{i_1\}$, we derive from $\alpha$-core that there is a voter $i_2\in S\setminus \{i_1\}$ such that $\alpha\cdot  u_{i_2}(p)\geq \frac{|S|-1}{n}$ as the group of voters $S\setminus \{i_1\}$ can otherwise benefit by deviating to $q$. By repeatedly applying this argument, it follows that there is an order $i_1,\dots, i_{|S|}$ of the voters in $S$ such that $\alpha \cdot u_{i_t}(p) \geq  \frac{|S|-(t-1)}{n} $ for each $t\in[|S|]$. We hence conclude that
\[
\sum_{i \in S} u_i(p) = \sum_{t = 1}^{|S|} u_{i_t}(p)
\geq \frac{1}{\alpha} \sum_{t = 1}^{|S|} \frac{|S|-(t-1)}{n}
= \frac{1}{\alpha} \sum_{t = 1}^{|S|} \frac{t}{n}
= \frac{1}{\alpha} \cdot \frac{|S| \cdot (|S|+1)}{2n}.
\]
This means that $ \frac{2\alpha}{|S|}\sum\limits_{i\in S} u_i(p)\geq \frac{|S|+1}{n} \geq \frac{|S|}{n}$. Since this bound holds for every group $S\subseteq N$ with $\bigcap_{i\in S} A_i \neq \emptyset$, $p$ satisfies $2\alpha$-AFS for $\mathcal A$.
\end{proof}

\section{Analysis of Existing Rules}
\label{sec:alg_comparison}

We will now analyze known decomposable rules, namely the Nash product rule
(\nash), the conditional utilitarian rule (\cut), the fair utilitarian rule (\fut), and the uncoordinated equal shares rule (\ues), with respect to monotonicity, our population consistency axioms, and their approximation ratios to AFS and core. We refer to the subsequent subsections for the definitions of these rules and to \Cref{tab:alg_comparison} for a summary of the results of our analysis. In particular, this table shows that only \nash satisfies AFS and core, but it violates monotonicity and only satisfies weak population consistency. Conversely, \ues satisfies monotonicity and even strong population consistency, but it is only a $\Theta(n)$-approximation for AFS and core.\footnote{For the upper bounds on the approximation ratio to AFS and core, we note that for all profiles $\mathcal A$ and all decomposable distributions $p$, it holds that $u_i(p)\geq \frac{1}{n}$ for all $i\in N$ and that $\frac{1}{|S|} \sum_{i\in S} u_i(p)\geq \frac{1}{n}$ for all $S\subseteq N$. It thus follows immediately that every decomposable distribution rule satisfies $n$-core and $n$-AFS.} Finally, \cut and \fut are also only $\Theta(n)$-approximations to AFS and core and violate even weak population consistency. These results demonstrate the need for new rules in order to simultaneously satisfy (approximate) fairness, monotonicity, and demanding population consistency notions. %We refer to \citet{BBPS21a} for the proofs regarding the monotonicity of \nash, \cut, and \fut.

 %, and to \Cref{sec:existingRules} for the proofs of the remaining statements in \Cref{tab:alg_comparison}.

%\footnote{We focus on decomposable rules due to our interest in fairness and for the sake of brevity. The well-known rules which fail decomposability (the Utilitarian rule and the Egalitarian rule) also fail to give a satisfactory approximation to fairness.} with respect to the axioms we have defined.

\subsection*{Nash product rule (\nash)}

The Nash product rule (\nash) selects a distribution $p$ maximizing the Nash welfare $\prod_{i\in N} u_i(p)$  \citep{BMS05a,FGM16a,ABM20a}.
Although the solution to this convex optimization problem may be irrational, it can be approximated efficiently and its solution is guaranteed to satisfy both AFS and core.
However, \citet{BBPS21a} showed that \nash fails monotonicity, and we will show next that \nash violates ranked population consistency. More specifically, we will demonstrate that even if a candidate receives the maximal share in two disjoint elections for which \nash returns distributions with identical distribution rankings, the candidate's share can decrease when combining the two elections.

\begin{proposition}\label{prop:Nash}
	\nash satisfies weak population consistency but fails ranked population consistency.
\end{proposition}
\begin{proof}
First, for showing that \nash satisfies weak population consistency, let $\mathcal A$ and $\mathcal A'$ denote two voter-disjoint approval profiles such that \nash returns for both profiles the same distribution~$p$. Moreover, let $N$ denote the set of voters corresponding to $\mathcal A$ and $N'$ denote the set of voters corresponding to $\mathcal A'$. By the definition of \nash , it holds that $p$ maximizes both $\prod_{i\in N} u_i(p)$ and $\prod_{i\in N'} u_i(p)$. This immediately implies that $p$ also maximizes $\prod_{i\in N\cup N'} u_i(p)$. Hence, \nash will choose $p$ also for $\mathcal A+\mathcal A'$. (Strictly speaking, we need to consider tie-breaking in case multiple distributions maximize the Nash social welfare for $\mathcal A$ or $\mathcal A'$, but it is easy to see that the argument still holds when imposing any consistent tie-breaking mechanism).

Next, to see that \nash fails ranked population consistency, consider the following two profiles.\smallskip

{
  \centering
  \begin{tabular} {@{}llllll@{}}
  % Our algorithms
  $\mathcal{A}$: & $3$: $\{a\}$ & $2$: $\{b,c\}$ & $2$: $\{c\}$ & $3$: $\{a,b\}$ \\
  $\mathcal{A}'$: & $1$: $\{a\}$ & $1$: $\{b\}$ & $1$: $\{c\}$ &$2$: $\{a,b\}$ & $4$: $\{a,c\}$.
  \end{tabular}\par
}\smallskip

By solving the corresponding convex programs, it can be shown that \nash chooses for $\mathcal{A}$ the lottery $p$ given by $p(a)=0.6$, $p(b)=0$, and $p(c)=0.4$ and for $\mathcal A'$ the lottery $q$ with $q(a)\approx 0.608$, $q(b)\approx0.157$, and $q(c)\approx0.235$. Hence, we have that ${\succsim^p}={\succsim^q}$. However, for $\mathcal A+\mathcal A'$, \nash chooses the lottery $r$ with $r(a)\approx 0.558$, $r(b)\approx0.137$, and $r(c)\approx0.305$. Thus, it holds that $p(a)>r(a)$ and $q(a)>r(a)$, which shows that ranked population consistency is violated even for the candidate that obtains the maximal share in two disjoint profiles.
\end{proof}

\begin{remark}
	In the proof of \Cref{prop:Nash} we show that the share of a candidate can decrease even if it has the maximal share in two profiles and the distribution rankings for these profiles agree. It can also be shown that the share of a candidate can significantly increase when combining two profiles for which \nash chooses distributions with the same distribution ranking, even if the considered candidate gets a share of $0$ in either of the two profiles. To see this, let $\mathcal A$ denote the profile where $50$ voters report $\{a\}$, $49$ voters report $\{b,c\}$, $49$ voters report $\{c\}$, and $50$ voters report $\{a,b\}$. For this profile, \nash selects the distribution $p$ with $p(a)=\frac{50}{99}$, $p(b)=0$, and $p(c)=\frac{49}{99}$. Next, let $\mathcal A'$ denote the profile where $1$ voter reports $\{a\}$, $1$ voter reports $\{c\}$ and $200$ voters report $\{a,b\}$. \nash chooses for this profile the distribution $q$ with $q(a)=\frac{201}{202}$, $q(b)=0$, and $q(c)=\frac{1}{202}$. It is straightforward to verify that the distribution rankings $\succsim^p$ and $\succsim^q$ coincide as {$a\succ^x c\succ^x b$} for both $x\in \{p,q\}$.
	Finally, in the combined profile $\mathcal A+\mathcal A'$ with $400$ voters, \nash selects the lottery $r$ with $r(a)=\frac{153}{300}$, $r(b)=\frac{97}{300}$, and $r(c)=\frac{50}{300}$, thus demonstrating another severe violation of RPC.
\end{remark}

\subsection*{Conditional Utilitarian Rule (\cut)}
First introduced by \citet{Dudd15a}, the conditional utilitarian rule (\cut) gives every voter control over a share of $\frac{1}{n}$ and assumes that each voter uniformly distributes this share across the candidates in $A_i$ with the highest approval scores.
More formally, we denote by $\text{\cut}(\mathcal{A}, i)=\arg\max_{x\in A_i} |N_x(\mathcal{A})|$ the subset of voter $i$'s approved candidates with maximal approval score.
% Then, \cut returns the distribution $p$ defined by $p(x)=\frac{1}{n}\sum_{i\in N^{\text{\cut}}_x(\mathcal{A})} \frac{1}{\text{\cut}(\mathcal{A},i)}$.
Then, \cut returns the distribution $p$ defined by $p(x)=\frac{1}{n}\sum_{i\in N} \frac{\mathbb{I}[x\in\text{\cut}(\mathcal{A},i)]}{|\text{\cut}(\mathcal{A},i)|}$ for all $x\in C$, where $\mathbb{I}[x\in\text{\cut}(\mathcal{A},i)]$ is the indicator function that is $1$ if $x\in \text{\cut}(\mathcal{A},i)$ and $0$ otherwise.
\cut is known to satisfy strategyproofness and thus also monotonicity \citep{BBP+19a}.
However, as we show in the next two propositions, \cut fails even weak population consistency, and it is only a $\Theta(n)$-approximation to AFS and core.

\begin{proposition}
	\cut fails weak population consistency.
\end{proposition}
\begin{proof}
	We consider the following two profiles $\mathcal A$ and $\mathcal A'$, both of which contain $10$ voters and $3$ candidates $C=\{a,b,c\}$.\medskip

	{
		\centering
		\begin{tabular} {@{}lllll@{}}
			% Our algorithms
			$\mathcal{A}$: & $2$: $\{a\}$& $4$: $\{a,c\}$ & $1$: $\{c\}$ & $3$: $\{b\}$ \\
			$\mathcal{A}'$: & $6$: $\{a\}$ &  $2$: $\{b\}$ & $1$: $\{b,c\}$ & $1$: $\{c\}$ \\
			% \bottomrule
		\end{tabular}\par
	}\medskip

	In both profiles, candidate $a$ is approved by the most voters, so every voter who approves $a$ sends his $\frac{1}{n}$ share to this candidate. Consequently, \cut assigns a share of $\frac{6}{10}$ to $a$ in both $\mathcal A$ and $\mathcal A'$. Now, in $\mathcal{A}$, all of the remaining four voters approve only a single candidate, so we immediately get that \cut chooses the distribution $p$ with $p(a)=\frac{6}{10}$, $p(b)=\frac{3}{10}$, and $p(c)=\frac{1}{10}$. On the other hand, for $\mathcal A'$, we note that three voters approve $b$ but only two voters approve $c$, so the single voter approving $\{b,c\}$ sends his share to $b$. So, the outcome for $\mathcal A'$ is again the distribution $p$ with $p(a)=\frac{6}{10}$, $p(b)=\frac{3}{10}$, and $p(c)=\frac{1}{10}$. Finally, for $\mathcal A+\mathcal A'$, \cut will first assign a share of $\frac{12}{20}=\frac{6}{10}$ to $a$ because it is approved by the most voters. However, in $\mathcal A+\mathcal A'$, $c$ is approved by $7$ voters whereas $b$ is only approved by $6$ voters. Hence, the voter approving $\{b,c\}$ sends his share to $c$ and $c$ gets a share of $\frac{3}{20}>\frac{1}{10}$. This proves that \cut violates weak population consistency as it chooses the same distribution for $\mathcal A$ and $\mathcal A'$ but not for $\mathcal A+\mathcal A'$.
\end{proof}

\begin{proposition}
	\cut does not satisfy $\alpha$-AFS or $\alpha$-core for any $\alpha<\frac{n}{2}-1$.
\end{proposition}
\begin{proof}
	Consider the following approval profile $\mathcal A$ for an even number of voters $n$, which are partitioned into two sets $S=\{1,\dots, \frac{n}{2}-1\}$ and $S'=\{\frac{n}{2},\ldots, n-2\}$,  and $\frac{n}{2}+1$ candidates $C=\{x^*, x_1, \dots, x_{n/2}\}$.
	\begin{itemize}
		\item Each voter $i\in S$ reports $A_i=\{x_i, x_{n/2}\}$.
		\item Each voter $i\in S'$ reports $A_i = \{x^*, x_1, \ldots, x_{n/2-1}\}$.
		\item Voters $n$ and $n-1$ report $A_{n-1} = A_n = \{x^*\}$.
	\end{itemize}
	Now, let $p$ denote the distribution returned by \cut for this instance.
	Candidate $x^*$ is approved by the most voters, so all these voters send their $\frac{1}{n}$ share to $x^*$ and $p(x^*)=\frac{1}{2}+\frac{1}{n}$. Next, we consider the group of voters $S$.
	Each voter $i\in S$ will allocate his $\frac{1}{n}$ share to $x_i$ since $|N_{x_i}(\mathcal A)| > |N_{x_{n/2}}(\mathcal A)|$ for all $i\in S$. Thus, $p({x_i})=\frac{1}{n}$ for all $i\in S$ and $p({x_{n/2}}) = 0$.
	Together, we have that
	\begin{align*}
	 	\frac{1}{|S|} \sum_{i\in S} u_i(p) &= \frac{1}{|S|} \sum_{i\in S} \sum_{x\in A_i} p(x)
		= \frac{1}{|S|} \cdot \frac{|S|}{n}
		= \frac{2}{n-2} \cdot \frac{|S|}{n}
	 \end{align*}
	Since $\bigcap_{i\in S} A_i = \{x_{n/2}\} \neq \emptyset$, this means that \cut does not satisfy $\alpha$-AFS for any $\alpha < \frac{n}{2}-1$.

	Next, for core, consider the distribution $q$ with $q(x_{n/2})=1$. Then, for all $i\in S$,
	\begin{align*}
		\frac{u_i(q)}{u_i(p)}\cdot \frac{|S|}{n} = \frac{1}{1/n}\cdot \frac{n/2-1}{n} = \frac{n}{2}-1.
	\end{align*}
	It follows that \cut cannot satisfy $\alpha$-core for any $\alpha<\frac{n}{2}-1$.
\end{proof}

\subsection*{Fair Utilitarian Rule (\fut)}

First introduced by \citet{BMS02a} and later rediscovered by \citet{BBPS21a}, \fut dynamically constructs weights for each voter and returns a distribution which maximizes the resulting weighted utilitarian welfare.
In more detail, the voters start with unit weights $\lambda_i = 1$ for all $i\in N$.
Then, \fut identifies the set of candidates $X_1$ that maximize $\sum_{i\in N_x(\mathcal A)} \lambda_i$ and sets $t=\sum_{i\in N_x(\mathcal A)} \lambda_i$.
Every voter $i$ with $A_i\cap X_1\neq\emptyset$ distributes their budget of $\frac{1}{n}$ uniformly among the candidates in $A_i\cap X_1$ and their $\lambda_i$ is fixed.
Then, the weights $\lambda_i$ of all other voters increase at a common rate until a candidate obtains a score of $t$, i.e., until $\sum_{i\in N_{x}(\mathcal A)} \lambda_i=t$ for some $x\in C\setminus X_1$. We then identify the set of candidates $X_2\subseteq C\setminus X_1$ which each have a total weight of $t$, and every voter $i\in N$ who did not spend his budget yet and approves at least one candidate in $X_2$ uniformly distributes his $\frac{1}{n}$ share among the candidates $A_i\cap X_2$. \fut then again increases the weights of voters who did not spend their $\frac{1}{n}$ share and repeats this process until all voters have allocated their share to some candidates.

Just like \cut, \fut satisfies monotonicity \citep[][]{BBPS21a}, but fails weak population consistency and only approximates AFS and core within a factor in $\Theta(n)$.

\begin{proposition}
	\fut fails weak population consistency.
\end{proposition}
\begin{proof}
We consider the following two profiles $\mathcal A$ and $\mathcal A'$, both of which contain $14$ voters and $3$ candidates $C=\{a,b,c\}$.\medskip

{
	\centering
	\begin{tabular} {@{}llllll@{}}
		% Our algorithms
		$\mathcal{A}$: & $3$: $\{a\}$& $5$: $\{a,c\}$ & $1$: $\{c\}$ & $1$: $\{b,c\}$ & $4$: $\{b\}$ \\
		$\mathcal{A}'$: & $6$: $\{a\}$ &  $2$: $\{a,b\}$ & $4$: $\{b\}$ & $2$: $\{c\}$ \\
		% \bottomrule
	\end{tabular}\par
}\medskip

For $\mathcal A$, \fut first assigns a share of $\frac{8}{14} $ to $a$ as this candidate is approved by $8$ voters. We then start increasing the weight of the remaining $6$ voters and note that the total score of $c$ will reach $8$ when $\lambda_i=\frac{3}{2}$ and the total score of $b$ is $8$ when $\lambda_i=\frac{8}{5}$. Hence, \fut next assigns a share of $\frac{2}{14}$ to $c$ and the remaining $\frac{4}{14}$ are assigned to $b$. This means that \fut chooses  for $\mathcal A$ the distribution $p$ with $p(a)=\frac{8}{14}$, $p(b)=\frac{4}{14}$, and $p(c)=\frac{2}{14}$. Next, for $\mathcal A'$, \fut again assigns a share of $\frac{8}{14}$ to $a$. Since all remaining voters only approve a single candidate and \fut is decomposable, we derive that \fut chooses for $\mathcal A'$ again the distribution $p$ with $p(a)=\frac{8}{14}$, $p(b)=\frac{4}{14}$, and $p(c)=\frac{2}{14}$. Finally, in the joint profile $\mathcal{A}+\mathcal{A}'$, \fut again assigns a share of $\frac{8}{14}$ to $a$ since it is approved by $16$ of $28$ voters. Next, we start increasing the weights and note that $c$ reaches a total weight of $16$ when $\lambda_i=\frac{11}{4}$ and $b$ when $\lambda_i=\frac{14}{9}$. Hence, the $9$ voters approving $b$ now send their portion to this candidate, so $b$ is assigned a share of $\frac{9}{28}>\frac{4}{14}$. This proves that \fut fails weak population consistency because it chooses the distribution $p$ for both $\mathcal A$ and $\mathcal A'$ but not for $\mathcal A+\mathcal A'$.
\end{proof}

\begin{proposition}\label{thm: FUT}
	\fut  does not satisfy $\alpha$-AFS or $\alpha$-core for any $\alpha < \frac{n}{3}-1$
\end{proposition}
\begin{proof}
	Assume that $n$ is divisible by $3$ and consider the following approval profile $\mathcal {A}$ for the candidates $C=\{x^*, x_1,\dots, x_{n/3-1},y\}$.
	\begin{itemize}
		\item Each voter $i\in \{1,\dots, \frac{n}{3}-1\}$ reports $\{x^*, x_i\}$.
		\item Voters $\frac{n}{3}$ and $\frac{n}{3}+1$ report $\{y\}$.
		\item Each voter $i\in \{\frac{n}{3}+2,\dots, n\}$ reports $\{ x_1,\dots, x_{n/3-1},y\}$.
		\end{itemize}
Observe that $|N_y(\mathcal A)|=\frac{2n}{3}+1>|N_x(\mathcal A)|$ for all $x\in C\setminus \{y\}$. Hence, \fut selects candidate $y$ first and sets $p(y)= \frac{2}{3}+\frac{1}{n}$. Moreover, the weights of the voters in $N_y(\mathcal A)$ are now fixed at $1$. Next, the weights of all other agents uniformly increase until some candidate has score of $\frac{2n}{3}+1$.
This occurs when $\lambda_i=2$ for all $i\in N\setminus N_y(\mathcal A)$ because $\sum_{i\in N_{x_i}(\mathcal A)}\lambda_i = \frac{2n}{3}-1+\lambda_i$ for each $x_i\in \{x_1,\dots,x_{x/3-1}\}$ and $\sum_{i\in N_{x^*}(\mathcal A)} \lambda_i=(\frac{n}{3}-1)\lambda_i$.
Thus, each candidate $x_i$ is assigned a share of $p({x_i})=\frac{1}{n}$ and $x^*$ is not assigned any budget.
Since all voter in $S=\{1,\dots, \frac{n}{3}-1\}$ approve $x^*$, we have that
\begin{align*}
\frac{1}{|S|}\sum_{i\in S} u_i(p) &= \frac{1}{\frac{n}{3}-1}\cdot (\frac{n}{3}-1)\cdot\frac{1}{n}=\frac{1}{n}.
\end{align*}
On the other hand, $\alpha$-AFS requires that $\frac{1}{|S|}\sum_{i\in S} u_i(p)\geq \frac{1}{\alpha}\cdot \frac{|S|}{n}=\frac{1}{\alpha}(\frac{1}{3}-\frac{1}{n})$. Hence, we conclude that \fut fails $\alpha$-AFS for every $\alpha<\frac{n}{3}-1$ when there are $n$ voters.

For the $\alpha$-core lower bound, we consider the distribution $q$ with $q(x^*)=1$. Then, we infer for every voter $i\in S$ that $u_i(p)=\frac{1}{n}$, but $\frac{|S|}{n}\cdot u_i(q)=\frac{1}{3}-\frac{1}{n}$, so \fut fails $\alpha$-core for every $\alpha<\frac{n}{3}-1$.
\end{proof}

\subsection*{Uncoordinated Equal Shares Rule (\ues)}

\ues assumes that each voter spends his $\frac{1}{n}$ share uniformly among his approved alternatives, i.e., it chooses the distribution $p$ given by $p(x) = \sum_{i\in N_x(\mathcal{A})} \frac{1}{n|A_i|}$ for all $x\in C$ \citep{MPS20a}. This rule is easily seen to satisfy both monotonicity and even strong population consistency. However, the approximation ratio of \ues to AFS and core is in $\Theta(n)$.

\begin{proposition}
\label{prop:es_afs}
\ues does not satisfy $\alpha$-AFS or $\alpha$-core for any $\alpha<\frac{n}{2}$.
\end{proposition}

\begin{proof}
	Consider the following approval profile $\mathcal A$ for $n$ voters and $1+n(n-1)$ candidates. Each voter $i$ approves candidate $x^*$, $n-1$ additional candidates, and the sets $A_i\setminus \{x^*\}$ are pairwise disjoint. More formally, when denoting the set of candidates $C$ by $C=\{x^*\}\cup \{x_j^i\colon i\in [n], j\in [n-1]\}$, then the approval ballot of each voter $i$ is given by $A_i=\{x^*\}\cup \{x_j^i\colon j\in [n-1]\}$.
	Now, let $p$ denote the distribution returned by \ues.
	It holds that $p(x^*) = \frac{1}{n}$ and that $p(x_i^j)=\frac{1}{n^2}$, which implies that $u_i(p)=\frac{2n-1}{n^2}$ for each voter $i\in N$. Consequently, we derive that
	$\frac{1}{n} \sum_{i\in N} u_i(p) =  \frac{2n-1}{n^2}$. On the other hand, $\alpha$-AFS requires that $\frac{1}{n} \sum_{i\in N} u_i(p)\geq\frac{1}{\alpha}$. Hence, it follows that \ues only satisfies $\alpha$-AFS for $\alpha\geq \frac{n^2}{2n-1}> \frac{n}{2}$.

	Moreover, it is easy to see that the utility of each agent improves by a factor $\frac{n^2}{2n-1}$ when moving from $p$ to the distribution $q$ with $q(x^*)=1$, which gives the lower bound for $\alpha$-core.
\end{proof}

%This does not seem very crucial
%\begin{remark}
%	The proof of \Cref{prop:es_afs} uses roughly $n^2$ candidates. By using more candidates, it is straightforward to adapt the proof to give higher bounds on the approximation ratio of \ues to $\alpha$-AFS and $\alpha$-core. Conversely, when we restrict ourselves to instances with $m\leq n$ candidates, one can still show a lower bound of $\Omega(\sqrt{n})$ for the approximation ratio of \ues to AFS and core.
%\end{remark}

%%%%%%%%%%%%%%%%%%%%%%%%%%%%%%%%%%%%%%%%%%%%%%%%%%%%%%%%%%%%%%%%%%%%%%%%
\section{Maximum Payment Rule}

In order to obtain more positive results, we will now introduce the \emph{maximum payment rule (\map)} and show that it satisfies monotonicity, ranked population consistency, and strong fairness guarantees. In combination with its simplicity, these properties make a strong case for using \map in practice. 

The idea of \map is to distribute the budget uniformly to the voters, who spend their shares sequentially on the candidates. In more detail, in every step, \map identifies the candidate $x$ that is approved by the maximum number of voters who have not yet spent their shares, and it allocates the budget of these voters to $x$. This process is repeated until all voters have assigned their share to a candidate and the final distribution corresponds to the amount of budget spent on each candidate. 
To make this more formal, we denote by $\Pi(\mathcal{A}, x, X)=|\{i\in N_x(\mathcal{A})\colon A_i\cap X=\emptyset\}|$ the number of voters who are willing to spend their share on $x$ after we assigned parts of the budget to the candidates in $X$. Now, letting $X$ denote the set of candidates that have been assigned budget in prior rounds, \map identifies in each round the candidate $x^*\in C\setminus X$ that maximizes $\Pi(\mathcal{A}, x, X)$, with ties broken lexicographically, and assigns a budget of $\frac{1}{n}\Pi(\mathcal{A}, x^*, X)$ to this candidate.

%We now provide a more formal description of \map. Suppose in a given round a set of candidates $X\subseteq C$ have been allocated their share in \map, then every voter who approve some candidate in $X$ has spent their entire budget. For any candidate $x\in C\setminus X$, let $\Pi(\mathcal{A}, x, X)$ denote the number of agents who approve candidate $x$ but do not approve any candidate in $X$, under the approval profile $\mathcal{A}$.  In the next round, \map selects a candidate $y = \argmax_{x\in C\setminus X} \Pi(\mathcal{A}, x, X)$, breaking ties lexicographically, and assigns to candidate $y$ a share of $\frac{1}{n}\Pi(\mathcal{A}, y, X)$. 

	\begin{example}
		To further illustrate \map, we compute its distribution for the profile $\mathcal{A}$ shown below.\medskip
		{
			\centering
			\begin{tabular} {@{}rccccc@{}}
				%& $\{a,b\}$ & $\{a\}$ & $\{b,c\}$ & $\{c,d\}$ & $\{d\}$ \\ \midrule
				% Our algorithms
				$\mathcal{A}$: & $4$: $\{a,b\}$ & $4$: $\{a\}$ & $2$: $\{b,c\}$ & $1$: $\{c,d\}$ & $1$: $\{d\}$
				% \bottomrule
			\end{tabular}\par
		}\medskip
		Since there are $n = 12$ voters, each agent starts with a budget of $\frac{1}{12}$. In the first round, all agents have non-zero budgets, so the first candidate selected is the approval winner—namely, candidate~$a$. All voters who approve of $a$ spend their entire budget to this candidate, resulting in $p(a) = \frac{8}{12}$. In the second round, the  candidate who is approved by the largest number of voters who did not spent their budget yet is $c$, leading to $p(c)=\frac{3}{12}$. In the final round, only the voter reporting $\{d\}$ has some budget remaining, so this voter sends his budget to $d$ and $p(d)=\frac{1}{12}$. 
		 Since all voters have now spent their entire budget, no part of the budget will be allocated to $b$ and \map selects the distribution $p$ with $p(a)=\frac{8}{12}$, $p(b)=0$, $p(c)=\frac{3}{12}$, and $p(d)=\frac{1}{12}$.
	\end{example}
	
	We next prove that \map satisfies monotonicity, ranked population consistency and strong fairness properties. We moreover note that \map is unanimous, thus demonstrating that RPC indeed allows to circumvent the impossibility in \Cref{prop:SPC}.
	
	\begin{theorem}\label{thm:MPmain}
		\map satisfies the following properties:
		\begin{enumerate}[leftmargin=*,label=(\arabic*)]
			\item Monotonicity.
			\item Ranked population consistency.
			\item  $2$-AFS and $\Theta(\log n)$-core. Moreover, these approximation ratios are tight.
			\end{enumerate}
		\end{theorem}
\begin{proof} 
	We prove each of the properties individually.\medskip
	
	\noindent\textbf{Monotonicity:}
	Let $\mathcal A$ and $\mathcal A'$ be two profiles, $i$ a voter, and $x^*$ a candidate such that $\mathcal A'$ is derived from $\mathcal A$ by adding $x^*$ to the approval ballot of voter $i$.
	We need to show that $q(x^*)\geq p(x^*)$ for the distributions $p=\text{\map}(\mathcal A)$ and $q=\text{\map}(\mathcal A')$. To this end, we observe for all profiles $\bar{\mathcal{A}}$, candidates $x\in C$, and sets $Y\subseteq C\setminus \{x\}$ that $\Pi(\bar{\mathcal{A}},x,Y)=|\{i\in N\colon x\in \bar A_i\land \bar A_i\cap Y=\emptyset\}|$.
	For our profiles $\mathcal A$ and $\mathcal{A}'$, this means that $\Pi(\mathcal A,x,Y)=\Pi(\mathcal A',x,Y)$ for all $x\in C\setminus \{x^*\}$ and $Y\subseteq C\setminus \{x, x^*\}$, and that $\Pi(\mathcal A,x^*,Y)\leq \Pi(\mathcal A',x^*,Y)$ for all $Y\subseteq C\setminus \{x^*\}$.
	Now, let $x_1,\dots, x_m$ and $x_1'\dots, x_m'$ denote the sequences according to which \map assigns the shares to the candidates in $\mathcal A$ and $\mathcal A'$, and let $\ell$ and $k$ denote the indices such that $x_\ell=x^*$ and $x_k'=x^*$.
	First, we observe that $\ell\geq k$ because the payment willingness for $x^*$ in $\mathcal{A}'$ is always weakly higher than in $\mathcal{A}$.
	Moreover, since $\Pi(\mathcal A,x,Y)=\Pi(\mathcal A',x,Y)$ for all $x\in C\setminus \{x^*\}$ and $Y\subseteq C\setminus \{x, x^*\}$, it follows that $x_i=x_i'$ for all $i<k$. This means that $\{x_1',\dots, x_{k-1}'\}=\{x_1,\dots, x_{k-1}\}\subseteq \{x_1,\dots, x_{\ell-1}\}$, so we conclude that
	\begin{align*}
		\Pi(\mathcal A',x^*,\{x_1',\dots, x_{k-1}'\})
		\geq \Pi(\mathcal A,x^*,\{x_1,\dots, x_{k-1}\})
		\geq \Pi(\mathcal A,x^*,\{x_1,\dots, x_{\ell-1}\}).%=p(x^*).
	\end{align*}
	%chg
	
	This shows that \map satisfies monotonicity because $q(x^*)=\frac{1}{n}\Pi(\mathcal A',x^*,\{x_1',\dots, x_{k-1}'\})$ and $p(x^*)=\frac{1}{n}\Pi(\mathcal A,x^*,\{x_1,\dots, x_{\ell-1}\})$.
	\medskip

	\noindent\textbf{Ranked population consistency:} 
	Let $\mathcal A$ and $\mathcal A'$ denote two voter-disjoint profiles, and let $p=\text{\map}(\mathcal A)$, $q=\text{\map}(\mathcal A')$, and $r=\text{\map}(\mathcal{A}+\mathcal A')$ be the distributions chosen by \map for $\mathcal{A}$, $\mathcal{A'}$, and $\mathcal{A}+\mathcal{A'}$. Moreover, let $z\in C$ denote a candidate such that ${\succsim^p}|_z={\succsim^q}|_z$. We need to show that $\min\{p(z), q(z)\}\leq r(z)\leq \max\{p(z),q(z)\}$. To this end, we will denote the strict part of a given distribution ranking $\succsim^s$ by $\succ^s$ (i.e., $x\succ^s y$ if and only if $s(x)>s(y)$) and the indifference part by $\sim^s$ (i.e., $x\sim^s y$ if and only if $s(x)=s(y)$). Furthermore, we let $x_1,\dots, x_m$ and $x_1',\dots, x_m'$ denote the sequences according to which \map assigns the shares to the candidates in $\mathcal A$ and $\mathcal A'$, respectively. Finally, let $k$ and $k'$ denote the index such that $x_k=z$ and $x_k'=z$. As a first step, we will show that $x_i=x_i'$ for all $i\in[k]$, which means that $k=k'$.
	For this, we observe for all $i\in \{1,\dots, m-1\}$ that
	\begin{align*}
		\Pi(\mathcal{A}, x_i,\{x_1,\dots, x_{i-1}\})\geq \Pi(\mathcal{A}, x_{i+1},\{x_1,\dots, x_{i-1}\})
		\geq \Pi(\mathcal{A}, x_{i+1},\{x_1,\dots, x_{i}\}).
	\end{align*}
	The first inequality holds because $x_i$ maximizes $\Pi(\mathcal{A}, x,\{x_1,\dots, x_{i-1}\})$ among all candidates in $C\setminus \{x_1,\dots, x_{i-1}\}$ and the second one as $\Pi(\mathcal A, y, Y)\geq \Pi(\mathcal A, y, Y')$ when $Y\subseteq Y'$. This means that $p(x_1)\geq p(x_2)\geq \dots \geq p(x_m)$ and we can analogously show that $q(x_1')\geq q(x_2')\geq \dots \geq q(x_m')$.
	
	Now, fix two indices $i,j\in [m]$ with $i>j\geq k$, which entails that $p(x_i)\geq p(x_j)\geq p(x_k)$. If $p(x_i)>p(x_j)$, then $x_i\succ^p x_j\succsim^p x_k$ by definition. Since ${\succsim^p}|_{x_k}={\succsim^q}|_{x_k}$, it thus follows that $x_i\succ^q x_{j}$ and so $q(x_i)>q(x_{j})$. By our previous insight, this means that $x_i$ is processed before $x_j$ for both $\mathcal{A}$ and $\mathcal{A}'$.
	On the other hand, if $p(x_i)=p(x_j)$, then $\Pi(\mathcal{A}, x_i,\{x_1,\dots, x_{i-1}\})= \Pi(\mathcal{A}, x_{j},\{x_1,\dots, x_{j-1}\})$. By using the definition of $x_i$ and the fact $\Pi(\mathcal A, y, Y)\geq \Pi(\mathcal A, y, Y')$ when $Y\subseteq Y'$, we infer from this that $\Pi(\mathcal{A}, x_i,\{x_1,\dots, x_{i-1}\})= \Pi(\mathcal{A}, x_{j},\{x_1,\dots, x_{i-1}\})$, too.
	Hence, $x_i$ is assigned its share of the budget before $x_j$ in $\mathcal A$ because of the lexicographic tie-breaking. Moreover, it holds by definition that $x_i\sim^p x_j\succsim^p x_k$. Since ${\succsim^p}|_{x_k}={\succsim^q}|_{x_k}$, this means that $x_i \sim^q x_j$ and thus $q(x_i)=q(x_j)$.
	Because our tie-breaking favors $x_i$ over $x_j$, this means again that $x_i$ is processed before $x_j$ in $\mathcal{A}'$. By combining our observations so far, we conclude that $x_i=x_i'$ for all $i\in [k]$.
	
	Next, we observe that $\Pi(\mathcal{A} + \mathcal{A}', x, Y) = \Pi(\mathcal{A}, x, Y) + \Pi(\mathcal{A}', x, Y)$ for all $x\in C$, $Y\subseteq C\setminus \{x\}$. Since $x_i=x_i'$ for all $i\in [k]$, it follows for all such $x_i$ and all $y\not\in \{x_1,\dots, x_{i-1}\}$ that
	\begin{align*}
		\Pi(\mathcal{A}+\mathcal{A}', x_i, \{x_1,\dots, x_{i-1}\}) &=\Pi(\mathcal{A}, x_i, \{x_1,\dots, x_{i-1}\})+\Pi(\mathcal{A}', x_i, \{x_1,\dots, x_{i-1}\})\\
		&\geq \Pi(\mathcal{A}, y, \{x_1,\dots, x_{i-1}\})+\Pi(\mathcal{A}', y, \{x_1,\dots, x_{i-1}\})\\
		&=\Pi(\mathcal{A}+\mathcal{A}', y, \{x_1,\dots, x_{i-1}\}).
	\end{align*}
	
	If this inequality is tight, it holds that $\Pi(\mathcal{A}, x_i, \{x_1,\dots, x_{i-1}\})=\Pi(\mathcal{A}, y, \{x_1,\dots, x_{i-1}\})$ and that $\Pi(\mathcal{A}', x_i, \{x_1,\dots, x_{i-1}\})=\Pi(\mathcal{A}', y, \{x_1,\dots, x_{i-1}\})$. We then infer that $x_i$ is lexicographically favored to $y$, so $x_i$ will be processed before $y$ in $\mathcal{A}+\mathcal{A}'$. It thus follows for the sequence $x_1'',\dots, x_m''$ according to which $f$ processes the candidates in $\mathcal{A}+\mathcal{A}'$ that $x_i''=x_i$ for all $i\in [k]$.
	
	Finally, assume there are $n$ voters in $\mathcal A$ and $n'$ in voters in $\mathcal A'$. Our insights so far imply that
	\begin{align*}
	r(z)&=\frac{1}{n+n'} \Pi(\mathcal A+\mathcal A', x_k, \{x_1,\dots, x_{k-1}\})\\
		&=\frac{1}{n+n'} (\Pi(\mathcal A, x_k, \{x_1,\dots, x_{k-1}\})+\Pi(\mathcal A', x_k, \{x_1,\dots, x_{k-1}\}))\\
		&=\frac{n}{n+n'}\cdot \frac{1}{n}\cdot \Pi(\mathcal A, x_k, \{x_1,\dots, x_{k-1}\}) + \frac{n'}{n+n'}\cdot \frac{1}{n'}\cdot \Pi(\mathcal A', x_k, \{x_1,\dots, x_{k-1}\})\\
		&=\frac{n}{n+n'} p(z) + \frac{n'}{n+n'} q(z).
	\end{align*}
	
	This shows that $r(z)$ is a convex combination of $p(z)$ and $q(z)$, so $\min\{p(z), q(z)\}\leq r(z)\leq \max\{p(z), q(z)\}$. This completes the proof that $f$ satisfies RPC.
	\medskip
	
	\noindent\textbf{$2$-AFS and $\Theta(\log n)$-core:} We  first show that \map satisfies $2$-AFS and give a matching lower bound. By \Cref{prop:afs-core}, it then follows that \map satisfies $O(\log n)$-core because every rule that satisfies $2$-AFS satisfies $2(1+\log n)$-core. Finally, we complete the proof by giving a family of profiles where \map only satisfies $\Omega(\log n)$-core.
	
	\medskip
	\noindent
	\textbf{Claim 1: \map satisfies $2$-AFS.}
	
	Fix some approval profile $\mathcal A$, and let $p=\text{\map}(\mathcal A)$ denote the distribution returned by \map.
	Moreover, let $S\subseteq N$ denote a subset of voters and $x^*\in C$ a candidate such that $x^*\in A_i$ for all $i\in S$. Lastly, let $x_1,\dots, x_m$ denote the sequence according to which \map assigns shares to the candidates, and let $k$ denote the index such that $x_k=x^*$.
	The total utility of the voters in $S$ is
	\begin{align*}
		\sum_{i\in S} u_i(p)=\frac{1}{n}\sum_{j\in [m]} |\{i\in S\colon x_j\in A_i\}| \cdot \Pi(\mathcal A, x_j, \{x_1,\dots, x_{j-1}\}).
	\end{align*}
	
	By the definition of \map, candidate $x_j$ maximizes $\Pi(\mathcal A, x, \{x_1,\dots, x_{j-1}\})$, so it holds for all $j\in [k]$ that $\Pi(\mathcal A, x_j, \{x_1,\dots, x_{j-1}\})\geq \Pi(\mathcal A, x^*, \{x_1,\dots, x_{j-1}\})$. Next, let $\hat S_j=\{i\in S\colon x_j\in A_i\land \{x_1,\dots, x_{j-1}\}\cap A_i=\emptyset \}$ denote the set of voters in $S$ for whom $x_j$ is the first approved candidate in our sequence, and note that $|\{i\in S\colon x_j\in A_i\}|\geq |\hat S_j|$. We conclude that
	\begin{align*}
		\sum_{i\in S} u_i(p)\geq \frac{1}{n}\sum_{j\in [k]}  |\hat S_j|  \cdot \Pi(\mathcal A, x^*, \{x_1,\dots, x_{j-1}\}).
	\end{align*}
	
	Next, it holds that $\Pi(\mathcal A, x^*, \{x_1,\dots, x_{j-1}\})\geq |S|- \sum_{\ell\in [j-1]} |\hat S_\ell|$ since the payment willingness of a voter is $0$ once one of his approved candidates has been assigned its share.
	This means that
	\begin{align*}
		\sum_{i\in S} u_i(p)\geq \frac{1}{n}\sum_{j\in [k]} |\hat S_j|\cdot (|S|- \sum_{\ell\in [j-1]} |\hat S_\ell|)= \frac{1}{n}\left( |S|\cdot\sum_{j\in [k]} |\hat S_j| - \sum_{j\in [k]} |\hat S_j| \sum_{\ell\in [j-1]} |\hat S_\ell|\right).
	\end{align*}
	
	We will next derive a lower bound on the right-hand term. To this end, we note that $\sum_{j\in [k]} |\hat S_j|=|S|$ because the sets $\hat S_j$ are disjoint and $\hat S_k=S\setminus \bigcup_{j\in [k-1]} \hat S_j$. Hence, it follows that $\sum_{j\in [k]} |\hat S_j|\cdot |S|=|S|^2$. 
	Next, it holds for all $\hat S_j$ that $|\hat S_j|\cdot \sum_{\ell\in [j-1]} |\hat S_\ell|=\sum_{r=0}^{|\hat S_j| - 1}\sum_{\ell\in [j-1]} |\hat S_\ell| \leq \sum_{r = 0}^{|\hat S_j| - 1} (r + \sum_{\ell \in [j-1]} |\hat S_\ell|)$. By applying this for all these sets, we derive that $\sum_{j\in [k]} |\hat S_j|\cdot \sum_{\ell\in [j-1]} |\hat S_\ell|\leq \sum_{j\in [k]} \sum_{r=0}^{|\hat S_j|-1}(r + \sum_{\ell \in [j-1]} |\hat S_\ell|) = \sum_{j=0}^{|S|-1} j =\frac{(|S|-1)|S|}{2}$. 
	By using these insights, we infer that
	\begin{align*}
		\sum_{i\in S} u_i(p)
		%\geq \frac{1}{n}\sum_{j\in [k]} |\hat S_j|\cdot (|S|- \sum_{\ell\in [j-1]} |\hat S_\ell|) \\
		\geq \frac{1}{n}\left( |S|\cdot\sum_{j\in [k]} |\hat S_j| - \sum_{j\in [k]} |\hat S_j| \sum_{\ell\in [j-1]} |\hat S_\ell|\right) 
		\geq  \frac{1}{n}\left(|S|^2-\frac{(|S|-1)|S|}{2}\right)
		\geq \frac{|S|^2}{2n}.
	\end{align*}
	
	Since this holds for all groups of voters $S$ that approve a common candidate, this proves that \map satisfies $2$-AFS.
	
	\medskip
	\noindent
	\textbf{Claim 2: \map fails $(2-\epsilon)$-AFS for every $\epsilon>0$.}
	
	Assume for contradiction that \map satisfies $(2-\epsilon)$-AFS for some $\epsilon>0$ and choose $\ell\in\mathbb{N}$ such that $(2-\epsilon) (\ell+3)<2\ell$. To derive a contradiction, we consider the following profile $\mathcal{A}$ with $\ell+1$ candidates $C=\{x_1,\dots, x_\ell, x^*\}$ and $n=\ell+\frac{\ell(\ell+1)}{2}$ voters.
	\begin{itemize}[]
		\item For each candidate $x_i\in \{x_1,\dots, x_\ell\}$, there is one voter who reports $\{x_i, x^*\}$. We will refer to this group of voters as $S$ and observe that $|S|=\ell$.
		\item For each candidate $x_i\in \{x_1,\dots, x_\ell\}$, there are $\ell+1-i$ voters who only approve $x_i$.
	\end{itemize}
	
	For this profile, \map will return the distribution $p$ defined by $p(x_i)=\frac{\ell+2-i}{n}$ for all candidates $x_i\in \{x_1,\dots, x_\ell\}$ and $p(x^*)=0$. In more detail, \map processes the candidates for $\mathcal A$ in the order $x_1,x_2,\dots, x_\ell, x^*$ because, for all $i\in [\ell]$, it holds for the total payment willingness that $\Pi(\mathcal A, x_i, \{x_1,\dots, x_{i-1}\})=\ell+2-i$ and $\Pi(\mathcal A, x^*, \{x_1,\dots, x_{i-1}\})=\ell+1-i$. Since each candidate $x_i\in C\setminus \{x^*\}$ is approved by exactly one voter in $S$, the total utility of these voters is
	\begin{align*}
		\sum_{i\in S} u_i(p)=\sum_{j\in [\ell]} \frac{\ell+2-j}{n}=
		\frac{1}{n} \sum_{j\in [\ell]} (j+1)=\frac{1}{n} \left(\ell+\frac{\ell(\ell+1)}{2}\right)=1.
	\end{align*}
	
	On the other hand, because all voters in $S$ approve $x^*$, $(2-\epsilon)$-AFS requires that
	\begin{align*}
		\sum_{i\in S} u_i(p)\geq \frac{1}{2-\epsilon}\cdot \frac{|S|^2}{n}= \frac{1}{2-\epsilon}\cdot \frac{\ell^2}{\ell+\frac{\ell(\ell+1)}{2}}= \frac{1}{2-\epsilon}\cdot \frac{2\ell}{\ell+3}>1.
	\end{align*}
	
	The last inequality follows by the choice of $\ell$. Since our two equations contradict each other, it follows that \map fails $(2-\epsilon)$-AFS.
	
	\medskip
	\noindent
	\textbf{Claim 3: \map only satisfies $\Omega(\log n)$-core.}
	
	To prove this claim, we will construct a family of approval profiles $\mathcal{A}^k$ with $n = k \cdot 2^{k-1}$ voters and $3\cdot 2^{k-1}-1$ candidates for all $k\in\mathbb{N}$ with $k\geq 2$ such that the distribution $p^k=\text{\map}(\mathcal{A}^k)$ only satisfies $\alpha$-core for $\alpha\geq \frac{k}{2}$. Since $\log_2 n = \log_2(k\cdot 2^{k-1})=k-1 + \log_2 k \leq 2k$, this shows that \map only satisfies $\Omega(\log n)$-core.
	
	Now, fix an integer $k\in \mathbb{N}$ with $k\geq 2$. To describe the profile $\mathcal{A}^k$, we will denote the set of voters of this profile by $N^k$ and we assume that $N^k=[k]\times [2^{k-1}]$. That is, every voter is indicated by a unique tuple $(i,j)$, where $i$ is best interpreted as ``row index'' and $j$ as ``column index''. Moreover, there are two types of candidates $X^k=\{x_1,\dots, x_{2^{k-1}}\}$ and $Y^k=\{y^i_\ell \colon i \in [k], \ell \in [2^{i-1}]\}$ and we assume that our tie-breaking prefers all candidates in $Y^k$ to those in $X^k$. The profile $\mathcal{A}^k$ is defined as follows:
	\begin{itemize}
		\item Every voter $(i,j)$ approves exactly one candidate in $Y^k$, namely the candidate $y^i_\ell$ for $\ell=\lceil\frac{j}{2^{k-i}}\rceil$. Less formally, the voters in the $i$-th row are partitioned into $2^{i-1}$ sets of size $2^{k-i}$ such that all voters in a set approve candidate~$y_\ell^i$.
		%Less formally, all voters in the first row (i.e., all voters of the form $(i,j)$ with $i=1$) approve $y_1^1$. %The first half of the voters in the second row approve $y_1^2$, and the second half of the voters in the second row approve of $y_2^2$. The voters in the third row are partitioned in four equal size sets that respectively approve the candidates $y_1^3,\dots, y_4^3$, and so on.
		\item Every voter $(k,j)$ only approves $x_j$ among all candidates in $X^k$. Less formally, the voters in the $k$-th row only approve the candidate in $X^k$ that matches their column entry.
		\item Every voter $(i,j)$ with $i<k$ approves $2^{k-i-1}$ candidates from $X^k$, namely all $x_\ell$ with $2^{k-i-1}\cdot {(\lceil\frac{j}{2^{k-i-1}}\rceil-1)} \leq\ell\leq 2^{k-i-1}\cdot \lceil \frac{j}{2^{k-i-1}}\rceil $. Less formally, for $i<k$, each voter in the $i$-th row approves $2^{k-i-1}$ candidates from $X^k$ and all candidates in $X^k$ are approved by $2^{k-i-1}$ voters of this row.
		%For instance, the first half of the voters in the first row approve $x_1,\dots,x_{2^{k-2}}$ and the second half of the voters in the first row approve the candidates in $
	\end{itemize}
	
	We note that the candidates in $Y^k$ are unique for every row and partition the rows into subgroups. By contrast, the candidates in $X^k$ cover the voters across all rows and voters in rows with smaller indices approve more candidates from $X^k$. An example of $\mathcal{A}^k$ for $k=3$ is shown below.
	
	\begin{center}
		\begin{tabular}{c|ccccccc}
			& $j=1$ & $j=2$ & $j=3$ & $j=4$\\\hline
			$i=1$ & $\{y_1^1, x_1,x_2\}$ & $\{y_1^1, x_1,x_2\}$ &  $\{y_1^1, x_3,x_4\}$ &  $\{y_1^1, x_3,x_4\}$\\
			$i=2$& $\{y_1^2, x_1\}$ & $1$: $\{y_1^2, x_2\}$ &$1$: $\{y_2^2, x_3\}$ & $1$: $\{y_2^2, x_4\}$ \\
			$i=3$&$\{y_1^3, x_1\}$ & $1$: $\{y_2^3, x_2\}$ &$1$: $\{y_3^3, x_3\}$ & $1$: $\{y_4^3, x_4\}$
			%$y^1_1, x_1, x_2,x_3,x_4$ & $y^1_1, x_1, x_2,x_3,x_4$ & $y^1_1, x_1, x_2,x_3,x_4$& $y^1_1, x_1, x_2,x_3,x_4$ \\
			%$y^1_1, x_5, x_6, x_7, x_8$ & $y^1_1, x_5, x_6, x_7, x_8$ & $y^1_1, x_5, x_6, x_7, x_8$ & $y^1_1, x_5, x_6, x_7, x_8$ \\
			%$y^2_1, x_1, x_2$ & $y^2_1, x_1, x_2$ & $y^2_1, x_3, x_4$ & $y^2_1, x_3, x_4$ \\
			%$y^2_2, x_5, x_6$ & $y^2_2, x_5, x_6$ & $y^2_2, x_7, x_8$ & $y^2_2, x_7, x_8$ \\
			%$y^3_1, x_1$ & $y^3_1, x_2$ & $y^3_2, x_3$ & $y^3_2, x_4$ \\
			%$y^3_3, x_5$ & $y^3_3, x_6$ & $y^3_4, x_7$ & $y^3_4, x_8$ \\
			%$y^4_1, x_1$ & $y^4_2, x_2$ & $y^4_3, x_3$ & $y^4_4, x_4$ \\
			%$y^4_5, x_5$ & $y^4_6, x_6$ & $y^4_7, x_7$ & $y^4_8, x_8$
		\end{tabular}
	\end{center}
	
	In the profile $\mathcal A^k$, each candidate $y_\ell^i$ is approved by $2^{k-i}$ voters and each candidate in $X^k$ is approved by $1 + \sum_{r = 1}^{k-1} 2^{k-1-r} = 2^{k-1}$ voters.
	%In the profile $\mathcal A^k$, each candidate $y_j^i$ is approved by $2^{k-i}$ voters and each candidate $x_i$ is approved by $1+\sum_{j=1}^{k-1} 2^{k-1-j}=2^{k-1}$ voters.
	By our tie-breaking, \map first assigns a share of $\frac{1}{k}$ to candidate $y_1^1$. After this step, the total payment willingness for the candidates in $y_\ell^i$ with $i\geq 2$ remains unchanged, and the total payment willingness for the candidates in $X^k$ is $1 + \sum_{r = 2}^{k-1} 2^{k-1-r} = 2^{k-2}$.
	Hence, by our tie-breaking, \map will assign next a share of $\frac{1}{2k}$ to both $y_1^2$ and $y_2^2$.
	After these two steps, the total payment willingness for every candidate in~$X^k$ is $1 + \sum_{r = 3}^{k-1} 2^{k-1-r} = 2^{k-3}$ and the total payment willingness for the candidates $y_\ell^i$ with $i \geq 3$ is still $2^{k-i}$.
	By repeating this reasoning, we derive that \map returns the distribution~$p^k$ given by $p^k(y_\ell^i) = \frac{1}{k \cdot 2^{i-1}}$ for all $y^i_\ell \in Y^k$ and $p^k(x) = 0$ for all~$x \in X^k$.
	Consequently, the utility of every voter $(i, j)$ for $p^k$ is $u_{(i, j)}(p^k) = \frac{1}{k \cdot 2^{i-1}}$.
	
	Now, consider the distribution $q$ defined by $q(x)=\frac{1}{2^{k-1}}$ for every candidate $x\in X^k$. Every voter $(i,j)$ with $i<k$ approves $2^{k-i-1}$ candidates in $X^k$, so $u_{(i,j)}(q)=\frac{1}{2^{k-1}} \cdot 2^{k-i-1}=\frac{1}{2^{i}}$. Moreover, the voters $(i,j)$ with $i=k$ approve a single candidate $x\in X^k$, so their utility is $u_{(i,j)}(q)=\frac{1}{2^{k-1}}$. Hence, it holds that $u_{(i, j)}(q) \geq \frac{k}{2} \cdot u_{(i, j)}(p^k)$ for every voter $(i, j)$, which shows that \map fails $\alpha$-core for every $\alpha< \frac{k}{2}$.
	\end{proof}

	\begin{remark}
		The profiles in Claim 3 of the proof of \Cref{thm:MPmain} show that \map violates efficiency because all voters can improve their utility by an $\Omega(\log n)$-factor in these profiles. 
		However, we note that our construction requires a large number of candidates $m$, and we can use a result by \citet{MPS20a} to limit the efficiency loss for small $m$. 
		Specifically, the proof of Theorem 2 of these authors demonstrates that the distribution chosen by \map guarantees a utilitarian social welfare of at least $\frac{2}{\sqrt{m}}-\frac{2}{m}$ of the optimum. 
		In turn, this implies that we can improve the utility of every voter by at most $\frac{\sqrt{m}}{2-2/\sqrt{m}}$ under \map. 
		Hence, for many realistic values of $m$, the efficiency loss of \map is not crucial. Despite this, it should be noted that \map can even violate efficiency if there are only three alternatives. 
	\end{remark}
	
	%\begin{remark}
	%	At a quick glance, \map may seem like a very utilitarian rule. While \map indeed maximizes the utilitarian social welfare among sequential payment rule in some important special cases
	%	(such as laminar profiles \citep{PeSk20a}), no such statement holds in general. For a counterexample, assume that there are $k$ voters approving $\{a^*,a_1,\dots, a_\ell\}$, two voters approving only $a^*$, and for each $i\in [\ell]$, there is one voter approving $\{a_i,b_i\}$ and one voter approving $\{b_i\}$. In this instance, \map will assign a share of $\frac{k+2}{n}$ to $a^*$ and a share of $\frac{2}{n}$ to each candidate $b_i$. However, if $k$ is large, assigning a share to $a_i$ instead of $b_i$ results in a higher social welfare, and every sequential payment rules other than \map will assign parts of the budget to the candidates $a_i$ if $k$ is sufficiently large. This shows that the approximation ratio of the utilitarian social welfare of \map can be arbitrarily bad and that it is not superior to that of other sequential payment rules.
%	\end{remark}
	
	\begin{remark}
		\label{para:tight}
		We note that \map shows that the implications between $\alpha$-AFS and $\alpha$-core in \Cref{prop:afs-core} are asymptotically tight. In more detail, as \map satisfies $2$-AFS and $\Theta(\log n)$-core, it is not possible to show a sub-logarithmic implication from approximate AFS to approximate core. Moreover, in the profiles used to show that \map fails $\alpha$-AFS for $\alpha<2$, the approximation ratio of \map for core converges to $1$, thus showing that it is not possible to show a better implication from approximate core to approximate AFS than the one given in \Cref{prop:afs-core}.
	\end{remark}
	
	\begin{remark}
		While \map is monotonic, it fails the more demanding property of strategyproofness. To see this, consider the following profile $\mathcal{A}$ with $7$ voters and $4$ candidates $C=\{a,b,c,d\}$.\medskip
		
		{\centering
			\begin{tabular} {@{}llllll@{}}
			%& $\{a,b\}$ & $\{a\}$ & $\{b,c\}$ & $\{c,d\}$ & $\{d\}$ \\ \midrule
			% Our algorithms
			$\mathcal{A}$: & $1$: $\{a,b,c\}$ & $2$: $\{a,b\}$ & $2$: $\{a,c\}$ & $1$: $\{b,d\}$ & $1$: $\{c,d\}$
			% \bottomrule
		\end{tabular}\par
	}\smallskip
		
		\noindent For this profile, \map chooses the distribution $p$ with $p(a)=\frac{5}{7}$ and $p(d)=\frac{2}{7}$. However, if the first voter reports $\{b,c\}$ instead of $\{a,b,c\}$ and the tie-breaking prefers $b$ over $a$ and $c$, the outcome changes to the distribution $q$ with $q(b)=\frac{4}{7}$ and $q(c)=\frac{3}{7}$. Since this voter has a utility of $\frac{5}{7}$ when voting honestly but of $1$ when reporting $\{b,c\}$, \map is manipulable.
	\end{remark}

\section{Sequential Payment Rules}\label{sec:SPR}
% \mashsays{@Jer we could use a slightly different intro to this section. Highlighting that SPR is a generalization of MP and that it is also an analogue of sequential thiele rules in this context}

Motivated by the positive results for \map, we will now generalize the underlying idea of this rule, which results in the class of \emph{sequential payment rules}. The idea of these rules is to allow voters to spend only parts of their shares on a candidate instead of allocating their full shares in one shot, as done by \map. In more detail, in every step of a sequential payment rule, each voter indicates a payment willingness for his approved candidates depending on his remaining budget. Just as for \map, we choose the candidate that maximizes the total payment willingness and the voters who are willing to pay for this candidate send the corresponding amount to this candidate. Finally, we remove this candidate from consideration, and repeat this process until all budget is spent.
%In addition to \map, the class of sequential payment rules also contains \ues, which we will discuss in more detail shortly.

%We will now introduce the central class of rules in this paper which we call sequential payment rules. The idea of these rules is that every voter has a virtual budget of $1$ and that voters sequentially spend this budget on the candidates. In particular, in each round, voters indicate, depending on their remaining budget, how much they are willing to spend on the next candidate in their approval set. The candidate that maximizes the total payment willingness is then chosen and the voters who are willing to pay for this candidate send the corresponding amount to this candidate. Finally, we remove this candidate from consideration, and repeat this process until all voters have spent all their budget.

%\textcolor{red}{PL: Shift second input for PWFs?}
Our key ingredient to formalize these rules are \emph{payment willingness functions} $\pi:\mathbb{N}\times [0,1]\rightarrow [0,1]$, which specify how much each voter is willing to spend on the next candidate depending on the number of his approved candidates and his remaining budget. We note %here that, for the definition of payment willingness functions, we assume that the budget of every voter is $1$; this makes it easier to handle profiles of variable sizes, and we can simply scale the budget and payment willingness of a voter down to $\frac{1}{n}$ by normalizing with $n$. Moreover,
that, from a technical point of view, payment willingness functions can also be seen as functions of the type $\pi:\mathbb{N}\times \mathbb{N}\rightarrow [0,1]$ as it suffices to know how many times a voter paid for candidates to infer his remaining budget and thus his next payment willingness. %In more detail, when defining $b_t$ as the remaining budget of a voter $i$ after he has spent parts of his budget on $t$ candidates, then $\pi(|A_i|, t+1)=\pi(|A_i|, b_t)$ and $b_{t+1}=b_t-\pi(|A_i|,t+1)$. By repeatedly applying this reasoning starting from $b_0=1$, we can substitute the remaining budget with the number of payments made by the voter. 
Since this description is easier to handle, we will subsequently treat payment willingness functions as mappings from $\mathbb{N}\times \mathbb{N}$ to $[0,1]$ and interpret $\pi(|A_i|, t)$ as the payment willingness of voter $i$ for his $t$-th candidate.
%\jvsays{If we're going to only use this format from now on, we should say so explicitly.}
Finally, payment willingness functions need to satisfy two more conditions:
\begin{enumerate}[leftmargin=*, label=(\arabic*)]
	\item Voters have to spend their whole budget: $\sum_{j=1}^{|A_i|} \pi(|A_i|, j)=1$ for all approval ballots $A_i$. Since $\pi(|A_i|, j)$ is always non-negative, this implies that the payment willingness of a voter is $0$ once he spent his entire budget.\smallskip
	\item The payment willingness is weakly decreasing in the budget: $\pi(|A_i|, j)\geq \pi(|A_i|, j+1)$ for all approval ballots $A_i$ and all integers $j\in [|A_i|-1]$.
	\end{enumerate}

Each sequential payment rule $f$ is fully specified by its payment willingness function $\pi$. In particular, when letting $X$ denote the set of candidates to whom we already allocated shares, the total payment willingness for a candidate $x\in C\setminus X$ is $\Pi(\mathcal{A}, x, X)=\sum_{i\in N_x(\mathcal{A})} \pi(|A_i|, |A_i\cap X|+1)$. A sequential payment rule then works as follows: starting from $X=\emptyset$, we identify in each round the candidate $x\in C\setminus X$ that maximizes $\Pi(\mathcal{A}, x, X)$ (with ties broken lexicographically), assign a share of $\frac{1}{n}\Pi(\mathcal{A}, x, X)$ to this candidate, and add $x$ to the set $X$. We then run this process for $m$ rounds to compute the final distribution. A pseudocode description of sequential payment rules is given in \Cref{alg:SPrules}.

\begin{algorithm}[t]
	\caption{Sequential payment rules}
	\label{alg:SPrules}
	
	\SetKwData{Left}{left}\SetKwData{This}{this}\SetKwData{Up}{up}
	\SetKwFunction{Union}{Union}\SetKwFunction{FindCompress}{FindCompress}
	\SetKwInOut{Input}{input}\SetKwInOut{Output}{output}
	
	\Input{An approval profile $\mathcal A$}
	\Output{A distribution $p\in \Delta(C)$}
	\BlankLine
	$X \leftarrow \emptyset$\;
	\For{$j\in [m]$}{
		$x \gets \argmax_{x \in C \setminus X} \sum_{i \in N_x(\mathcal{A})} \pi(|A_i|, |A_i \cap X| + 1)$\;
		$X \gets X \cup \{x\}$\;
		$p(x) \gets \frac{1}{n} \sum_{i \in N_x(\mathcal{A})} \pi(|A_i|, |A_i \cap X| + 1)$\;
	}
	\Return{$p$}
\end{algorithm}

To exemplify these definitions, we first note that \map is the sequential payment rule defined by the payment willingness function $\pi$ with $\pi(x,1)=1$ and $\pi(x,y)=0$ for all $x\in\mathbb{N}$ and $y\in [x]\setminus \{1\}$. Moreover, \ues also belongs to the class of sequential payment rules as it can be described by the payment willingness function given by $\pi(x,y)=\frac{1}{x}$ for all $x\in \mathbb{N}$, $y\in [x]$. Note here that, if the payment willingness function is constant in the second parameter, the order in which we pick the candidates does not matter. More generally, sequential payment rules can be seen as the counterpart of sequential Thiele rules, a well-studied family of voting rules for approval-based committee elections \citep{LaSk22b,DoLe23a}. However, we note that sequential payment rules use the payment willingness to assign shares to the candidates instead of identifying which candidate to add to the committee, which causes a significant difference in their behavior when compared to sequential Thiele rules.

As we will show next, \map stands out as the most appealing sequential payment rule. To this end, we first observe that all sequential payment rules but \map fail unanimity. This observation already rules out the usage of such rules for numerous applications. Furthermore, we show next that sequential payment rules do also not allow to significantly improve over \map with respect to the conditions studied in this paper. Specifically, we prove next that, while all sequential payment rules satisfy RPC, all such rules other than \map and \ues fail monotonicity. Moreover, every sequential payment rule violates $\alpha$-AFS for $\alpha<\frac{3}{2}$, so they also do not allow to substantially improve the fairness guarantees of \map. The proofs of the following theorem and of \Cref{thm:tightbound} have been deferred to the appendix.
%an online appendix \citep{TODO} upon request of the editor.

\begin{restatable}{theorem}{SPR}\label{thm:SPR}
	The following claims holds: 
\begin{enumerate}[leftmargin=*,label=(\arabic*)]
		\item All sequential payment rules satisfy RPC.
		\item All sequential payment rules except \ues and \map fail monotonicity.
		\item All sequential payment rules fail $\alpha$-AFS for $\alpha<\frac{3}{2}$. 
	\end{enumerate}
\end{restatable}

Despite the other drawbacks of sequential payment rules, \Cref{thm:SPR} suggests that they allow to improve the approximation ratio to AFS from $2$ to $\frac{3}{2}$. For applications where simplicity, population consistency, and fairness are the primary concerns (and outweigh efficiency and monotonicity violations), sequential payment rules other than \map may thus still be an option. Moreover, we find it an interesting question how well such simple combinatorial rules can approximate demanding fairness considerations. As our last contribution, we will therefore investigate which sequential payment rule optimizes the approximation ratio to AFS. 

To give a precise answer to this question, we will use a more fine-grained analysis that parameterizes the AFS approximation ratio based on the maximum ballot size in a profile. Specifically, we say a distribution rule satisfies \emph{$\alpha$-AFS for a maximal ballot size $t$} if it satisfies $\alpha$-AFS for all approval profiles where all voters approve at most $t$ candidates.
Based on this notion, we will single out the \emph{$\frac{1}{3}$-multiplicative sequential payment rule (\mps{\frac{1}{3}})} as the fairest sequential payment rule. The idea of this rule is that the payment willingness of a voter is always discounted by a factor of $\frac{1}{3}$ if he pays for a candidate. More formally, the payment willingness function of this rule is defined by the equations $\sum_{y\in [x]} \pi(x,y)=1$ for all $x\in \mathbb{N}$ and $\pi(x,y+1)=\frac{1}{3}\pi(x,y)$ for all $x\in\mathbb{N}$ and $y\in [x-1]$, which implies that $\pi(x,y)=\frac{3^{-y}}{\sum_{i\in [x]}3^{-i}}$ for all $x\in \mathbb{N}$, $y\in [x]$. As we show next, no sequential payment rule satisfies $\alpha$-AFS for a maximal ballot size $t\in\mathbb{N}$ if $\alpha<\frac{3}{2}(1-3^{-t})$ and \mps{\frac{1}{3}} is the only rule in this class that tightly matches this lower bound for all $t\in\mathbb{N}$. 

\begin{restatable}{theorem}{tightanalysis}\label{thm:tightbound}
	The following statements hold:
	\begin{enumerate}[leftmargin=*,label=(\arabic*)]
		\item Let $t\!\in\!\mathbb{N}$. No sequential payment rule satisfies $\alpha$-AFS for the maximal ballot size $t$ and ${\alpha\!<\! \frac{3}{2}(1-3^{-t})}$.
		\item \mps{\frac{1}{3}} is the only sequential payment rule that satisfies ${\frac{3}{2}(1-3^{-t})}$-AFS for all maximal ballot sizes $t\!\in\!\mathbb{N}$.
	\end{enumerate}
\end{restatable}

\begin{remark}
	\label{rem:SPR}
	We note that \mps{\frac{1}{3}} is a member of the broader class of $\gamma$-multiplicative sequential payment rules (\mps{\gamma}), where the payment willingness of each voter is discounted by $\gamma\in [0,1]$ (instead of $\frac{1}{3}$) after he pays for a candidate. In the appendix (\Cref{thm:MPS}), we determine the approximation ratio to AFS for every rule in this class as a function of $\gamma$. Further, while \mps{\frac{1}{3}} is the only sequential payment rule that satisfies ${\frac{3}{2}(1-3^{-t})}$-AFS for all maximal ballot sizes, there are multiple such rules that match the general lower bound of $\frac{3}{2}$. An example is the $\frac{1}{3}$-additive sequential payment rule, where the payment willingness of every agent decreases by $\frac{1}{3}$ whenever he pays for a candidate. More formally, this rule is defined by the payment willingness function $\pi$ with $\pi(1,1)=1$, $\pi(x,1)=\frac{2}{3}$, $\pi(x,2)=\frac{1}{3}$, and $\pi(x,y)=0$ for all $x\in \mathbb{N}\setminus \{1\}$, $y\in [x]\setminus \{1,2\}$. Additionally to satisfying $\frac{3}{2}$-AFS, this rule is monotonic for all instances where the voters are required to approve at least two candidates. This can be shown based on a similar argument as for \map since the payment willingness of each voter for his first two candidates (i.e., $\pi(x,1)$ and $\pi(x,2)$) does not change when approving an additional candidate.%Thus, this rule is an attractive alternative to \map if voters are required to approve more than one candidate.
	\end{remark}

\section{Conclusion}

In this paper, we study approval-based budget division with the aim of finding simple, consistent, and fair distribution rules. More specifically, our goal is to find simple combinatorial rules that satisfy monotonicity and strong population consistency conditions and, at the same time, give strong fairness guarantees. To this end, we first observe that all existing rules fail our desiderata. In particular, while \nash satisfies demanding fairness conditions such as core and AFS, it violates monotonicity and is a rather intrinsic rule. On the other hand, all other common rules (i.e., \cut, \fut, and \ues) give poor approximations to AFS and core, and \cut and \fut additionally fail weak forms of population consistency. Motivated by these observations, we suggest the maximum payment rule (\map), which determines the outcome by first distributing the budget uniformly to the voters and then letting the voters spend their shares on the candidates. In more detail, under \map, we repeatedly identify the candidate that is approved by the most voters who have not spent their share yet and assign the shares of these voters to the candidate. As we show, \map satisfies all of our conditions: it satisfies monotonicity as well as a demanding population consistency notion called ranked population consistency, and it gives a $2$-approximation to AFS and a $\Theta(\log n)$-approximation to core. 

Motivated by these positive results, we further generalize the underlying idea of \map to a family of rules called sequential payment rules. Intuitively, these rules differ from \map by allowing for more fine-grained payments: in each step, the voters now indicate a payment willingness for their approved candidates depending on their remaining budget, and we let voters send their corresponding shares to the candidate that maximizes the total payment willingness. We then prove that the class of sequential payment rules only offers limited improvements over \map. In more detail, while all sequential payment rules satisfy ranked population consistency, \map and \ues are the only rules within this class that satisfy monotonicity. Moreover, we prove that no sequential payment rule is an $\alpha$-approximation to AFS for $\alpha<\frac{3}{2}$, and we characterize the $\frac{1}{3}$-multiplicative payment rule (\mps{\frac{1}{3}}) as the fairest rule in this class.

%We show that all sequential payment rules satisfy a strong form of population consistency, which we call ranked population consistency. Moreover, we identify two particularly appealing rules within our class: the maximum payment rule (\map), where voters are willing to spend their whole budget on their first approved candidate, and the $\frac{1}{3}$-multiplicative payment rule (\mps{\frac{1}{3}}), where the payment willingness of voters is always discounted by a factor of $\frac{1}{3}$ when they pay for a candidate. In more detail, we show that \map is, except for the uncoordinated equal shares rule, the only sequential payment rule that satisfies monotonicity. Moreover, it guarantees a 2-approximation to a fairness notion called average fair share (AFS), and a $\Theta(\log n)$-approximation to core, thus demonstrating that \map is indeed a simple, well-behaved, and fair rule. On the other hand, we identify \mps{\frac{1}{3}} as the fairest sequential payment rule by showing that this rule is a $\frac{3}{2}$-approximation to AFS and that no other sequential payment has a better approximation ratio. We refer to \Cref{tab:alg_comparison} for an overview of our results and a comparison to existing rules.

Our work points to several interesting follow-up questions. For instance, a straightforward drawback of sequential payment rules is that they fail Pareto efficiency, leading to the question of whether one can define similar spending dynamics that are approximately fair and efficient. More generally, it seems interesting to study other types of spending dynamics that, e.g., could take the current utility of voters into account or allow voters to pay multiple times for a candidate.

\section*{Acknowledgments}

This work was partially supported by the NSF-CSIRO grant on ``Fair Sequential Collective Decision-Making'' and the ARC Laureate Project FL200100204 on ``Trustworthy AI''.
% XL on 2024-09-27: UPDATE ``Acknowledgements.''

%-----------------------------------------------------------------------------%
%	BIBLIOGRAPHY
%-----------------------------------------------------------------------------%
%\bibliographystyle{plainnat}
%\bibliography{./group}

%-----------------------------------------------------------------------------%
%	APPENDIX
%-----------------------------------------------------------------------------%
\clearpage
\appendix

%\input{rule_comparison}
% !TEX root = ssr_geb_ver2.tex

\section{Omitted Proofs From Section~5}

In this appendix, we present the omitted proofs of the claims from Section~\ref{sec:SPR}. We start by proving \Cref{thm:SPR}.

\SPR*
\begin{proof}
	We start by considering the three claims in \Cref{thm:SPR}. To this end, we first note that an analogous argument as for \map shows that every sequential payment rule satisfies RPC; we thus omit this proof. Secondly, the claim on $\alpha$-AFS is an immediate consequence of \Cref{thm:tightbound} by considering the limit for $t\rightarrow \infty$ of our lower bound. We will thus only show the claim on monotonicity here. 
	
	To this end, let $f$ denote a sequential payment rule other than \map and \ues and let $\pi$ denote its payment willingness function. Since $f$ is not \map, there is an integer $\ell_1$ such that $\pi(\ell_1,1)=1> \pi(\ell_1 +1,1)$.
	%Moreover, because $\sum_{h=1}^{\ell_1} \pi(\ell_1, h)=1\geq \sum_{h=1}^{\ell_1} \pi(\ell+1, h)$, we can assume that $\pi(\ell_1, k_1)>\pi(\ell_1+1,k_1)$; such an index is guaranteed to exist since otherwise $\pi(\ell_1, k)\leq \pi(\ell_1+1, k)$ for all $k\in [\ell_1]$, which implies that $\pi(\ell_1, k)= \pi(\ell_1+1, k)$ for all $k\in [\ell_1]$.
	Moreover, since $f$ is not \ues, there is another pair of indices $\ell_2, k$ with $k\in [\ell_2-1]$ such that $\pi(\ell_2,k)\neq \pi(\ell_2,k+1)$. Because payment willingness functions are non-increasing in their second parameter, this means that  $\pi(\ell_2,k)> \pi(\ell_2,k+1)$. Based on these parameters, we will now construct a profile where $f$ fails monotonicity. To this end, let $\ell=\max(\ell_1,\ell_2)$. The rough idea of our construction is as follows: there will be three ''active'' candidates $x,y,z$ and $2\ell-2$ inactive candidates. In the original profile, $f$ will process the active candidates in the order $y,x,z$. By contrast, if some voters additionally approve $x$, this will decrease the payment willingness for $y$ as $\pi(\ell_1, 1)>\pi(\ell_1+1,1)$. In particular, $z$ will then be chosen as the first candidate among $x$, $y$, and $z$. However, choosing $z$ before $x$ will reduce the payment willingness towards $x$ as $\pi(\ell_2, k)>\pi(\ell_2, k+1)$, thus resulting in a smaller share for $x$ despite the fact that it gained more approvals. To formalize this, we will use $2\ell+1$ candidates which are partitioned into three groups $B=\{b_1,\dots, b_{\ell-1}\}$, $D=\{d_1,\dots, d_{\ell-1}\}$, and $x,y,z$. For convenience, we let by $B(s) = \{b_1, \dots, b_s\}$ and $D(s) = \{d_1, \dots, d_s\}$. Then, we consider the following approval profile $\mathcal A$.
	\begin{itemize}
		\item Let $\delta_1=\pi(\ell_1,1)-\pi(\ell_1+1,1)$ and define $t_1$ as an integer such that $t_1 \delta_1>1$. We add $t_1$ voters approving the set $D(\ell_1-1) \cup \{y\}$ and we call this set of voters~$N_1$.
		\item Let $\delta_2=\pi(\ell_2, k)-\pi(\ell_2,k+1)$ and choose an integer $t_2$ such that $t_2\delta_2>t_1 \pi(\ell_1+1,1)$. We add $t_2$ voters who report the set $B(k-1)\cup D(\ell_2-k-1)\cup \{x,z\}$.
		\item Choose an integer $t_3$ such that $t_3\delta_2>t_1\pi(\ell_1+1,1)+1$. We add $t_3$ voters who approve the set $B(k-1)\cup D(\ell_2-k-1)\cup \{y,z\}$.
		\item For each $w\in \{x,y,z\}$, we add at least $t_1+t_2+t_3$ voters who only approve $w$ to guarantee that the candidates in $\{x,y,z\}$ will be assigned their shares before the candidates in $D$. This works because the definition of payment willingness functions require that $\pi(1,1)=1$. Moreover, we can add more such voters to ensure that $\Pi(\mathcal {A}, z, B)+1\geq \Pi(\mathcal {A}, y, B)>\Pi(\mathcal {A}, z, B)$ and $\Pi(\mathcal {A}, x, B)+t_1\cdot \pi(\ell_1+1,1)+1\geq \Pi(\mathcal A, z, B)>\Pi(\mathcal A, x, B)+t_1\cdot \pi(\ell_1+1,1)$.
		\item Let $T$ denote the total number of voters introduced so far. For each $w\in B$, we add $2T$ voters who only approve $w$. This ensures that these candidates will be assigned their shares before all other candidates.
	\end{itemize}
	
	By construction, $f$ will first assign shares to the candidates in $B$ for the profile $\mathcal A$. Moreover, since we add sufficiently many voters who only approve $y$, this candidate will be processed next. Let $n_x$ and $n_z$ denote the number of voters who report $\{x\}$ and $\{z\}$, respectively. Now, it holds for the payment willingness for $x$ and $z$ that
	\begin{align*}
		\Pi(\mathcal A,z,B\cup \{y\})&=t_2\cdot \pi(\ell_2,k) + t_3 \cdot \pi(\ell_2, k+1)+n_z\\
		&=t_2\cdot \pi(\ell_2,k) + t_3 \cdot \pi(\ell_2, k)-t_3\delta_2+n_z\\
		&= \Pi(\mathcal A,z,B)-t_3\delta_2 \\
		&<\Pi(\mathcal A,z,B)-t_1\pi(\ell_1+1,1)-1\\
		&\leq\Pi(\mathcal A,x,B)\\
		&=\Pi(\mathcal A,x,B\cup \{y\}).
	\end{align*}
	
	Here, the first equality follows from the definition of $\Pi$, the second one uses the definition of $\delta_2$, and the third equation uses that $\Pi(\mathcal A, z, B)=t_2\cdot \pi(\ell_2,k) + t_3 \cdot \pi(\ell_2, k)+n_z$. The next line follows from the definition of $t_3$ and the fifth line because we ensure that $\Pi(\mathcal A,z,B)\leq \Pi(\mathcal A,x,B)+t_1\pi(\ell_1+1,1)+1$. Finally, the last equation holds since there is no voter in $\mathcal{A}$ that approves both $x$ and $y$. Because of this inequality, we now infer that $f$ next assigns a share of $\frac{t_2 \cdot \pi(\ell_2,k) + n_x}{n}$ to $x$.
	
	Now, let $\mathcal A'$ denote the profile where all voters in $N_1$ additionally approve $x$. We will show that the share of $x$ decreases, which contradicts monotonicity. To this end, we note that in $\mathcal A'$, the candidates in $B$ are still the first that will be assigned their shares by $f$. After this, $f$ assigns a share to $z$. In more detail, the payment willingness for $z$ is higher than that for $y$ because
	\begin{align*}
		\Pi(\mathcal A',y, B)&=t_1\cdot \pi(\ell_1+1,1)+t_3\cdot\pi(\ell_2,k)+n_y\\
		&=\Pi(\mathcal A, y, B)-t_1\delta_1\\
		&<\Pi(\mathcal A, y, B)-1\\
		&\leq \Pi(\mathcal A, z, B)\\
		&=\Pi(\mathcal A', z, B).
	\end{align*}
	
	Here, we use the definitions of $\Pi$, $\delta_1$, $t_1$, and that  $\Pi(\mathcal A, z, B)+1\geq  \Pi(\mathcal A, y, B)$ in this order. The final equation holds since all voters who approve $z$ report the same preferences in $\mathcal A$ and~$\mathcal A'$.
	
	Moreover, using similar reasoning, we also infer that the payment willingness for $z$ in $\mathcal A'$ is higher than that for $x$ because
	\begin{align*}
		\Pi(\mathcal A', x, B)&=t_1\cdot \pi(\ell_1+1,1)+t_2\cdot\pi(\ell_2,k)+n_x\\
		&=\Pi(\mathcal A, x, B)+t_1\cdot \pi(\ell_1+1,1)\\
		&<\Pi(\mathcal A, z, B)\\
		&=\Pi(\mathcal A', z, B).
	\end{align*}
	
	Hence, we indeed first choose $z$ among $\{x,y,z\}$.
	Now, assume that $x$ is chosen in the next step; if $y$ is chosen next, the share of $x$ only decreases further. The total payment willingness for $x$ is
	\begin{align*}
		\Pi(\mathcal A',x,B\cup \{z\})&=t_1\cdot \pi(\ell_1+1, 1)+t_2\cdot \pi(\ell_2,k+1)+n_x\\
		&=t_1\cdot \pi(\ell_1+1, 1)+t_2\cdot \pi(\ell_2,k)-t_2\delta_2+n_x\\
		&<t_2\cdot \pi(\ell_2,k)+n_x.
	\end{align*}
	The last inequality follows because we choose $t_2$ such that $t_2\delta_2>t_1 \pi(\ell_1+1,1)$. This shows that the share of $x$ in $\mathcal{A}'$ is less than $\frac{t_2 \cdot \pi(\ell_2,k) + n_x}{n}$, so $f$ fails monotonicity. 
\end{proof}

We next turn to the analysis of the approximation ratio of sequential payment rules to AFS. To facilitate the proof of the lower bound for this result, we will next present necessary conditions on the payment willingness function of every sequential payment rule that satisfies $\alpha$-AFS. More specifically, the next proposition directly relates the payment willingness function $\pi$ of a sequential payment rule with its approximation ratio $\alpha$ to AFS.

\begin{restatable}{proposition}{lowerbound}
	\label{prop:lowerbound}
	Let $f$ denote a sequential payment rule that satisfies $\alpha$-AFS for some $\alpha\in \mathbb{R}$, and let $t\in\mathbb{N}$. It holds for the payment willingness function $\pi$ of $f$ that
	\begin{align*}
		\pi(t,1)&\geq  \frac{1}{\alpha}\qquad\\
		\frac{1}{2}\pi(t,1) + \frac{3}{2} \pi(t,2)&\geq \frac{1}{\alpha}\\
		\frac{1}{2}\pi(t,1) + \pi(t,2) + \frac{3}{2} \pi(t,3)&\geq  \frac{1}{\alpha}\\
		&\dots\\
		\frac{1}{2}\pi(t,1) + \pi(t,2) + \dots + \pi(t,t-1)+\frac{3}{2}\pi(t,t)&\geq  \frac{1}{\alpha}.
	\end{align*}
\end{restatable}
\begin{proof}
	Let $f$ denote a sequential payment rule that satisfies $\alpha$-AFS for some finite $\alpha$, let $\pi$ be its payment willingness function, and let $t\in\mathbb{N}$ denote an arbitrary integer. Throughout this proof, we will interpret $t$ as the maximal feasible ballot size to ensure that we can use this proposition also in the proof of \Cref{thm:tightbound}. We moreover define $\pi_i=\pi(t,i)$ for all $i\in [t]$ to simplify notation and note that $\pi_1\geq \pi_2\geq \dots \geq \pi_t$ by the definition of payment willingness functions.
	
	In this proof, it will be easier to discuss fractional profiles as this allows us to avoid complicated computations. Formally, a {fractional profile} $W$ is a function from the set of approval ballots ${2^C\setminus \{\emptyset\}}$ to $\mathbb{Q}_{\geq 0}$, with the interpretation that $W(A)$ indicates how many voters report the ballot $A$ in the profile $W$. That is, rather than having $x\in\mathbb{N}_0$ voters reporting an approval ballot $A$, it is now possible that $x$ is a rational number.
	The size of a fractional profile $W$ is defined by $|W|=\sum_{A\in 2^C\setminus \{\emptyset\}} W(A)$ and we will assume that $|W|>0$ for all fractional profiles. It is straightforward to generalize $f$ to fractional profiles: we can extend the definition of the total payment willingness $\Pi$ to fractional profiles by letting $\Pi(W,x,Y)=\sum_{A\in 2^C\setminus \{\emptyset\}\colon x\in A} W(A) \cdot \pi(|A|, |A\cap Y|+1)$ for every fractional profile $W$, candidate $x$, and set of candidates $Y\subseteq C\setminus \{x\}$. Then, sequential payment rules work on fractional profiles as described in \Cref{alg:SPrules}. From the definition of sequential payments rules, it follows that $f(W)=f(\ell W)$ for every $\ell\in\mathbb{N}$, where the fractional profile $W'=\ell W$ is defined by $W'(A)=\ell\cdot W(A)$ for all approval ballots $A$. 
	
	Next, we extend the definition of $\alpha$-AFS to fractional profiles as follows: a distribution $p$ satisfies $\alpha$-AFS for a fractional profile $W$ if
	\[
	\frac{\alpha}{|S|} \cdot \sum_{A\in 2^C\setminus \{\emptyset\}} \bigg( S(A) \cdot \sum_{y\in A} p(y) \bigg) \geq  \frac{|S|}{|W|}
	\]
	for all fractional (sub)profiles $S$ and candidates $x\in C$ such that $S(A)\leq W(A)$ for all $A\in\mathcal{A}$ and $S(A)>0$ implies that $x \in A$.
	%	for all fractional (sub)profiles $S$ and candidates $x\in C$ such that $S(A)\leq W(A)$ for all $A\in\mathcal{A}$ and $S(A)>0$ implies that $x\in C$.
	Intuitively, $S$ can be seen here a set of (fractional) voters such that $x\in A_i$ for all voters $i\in S$.
	It holds that $f$ satisfies $\alpha$-AFS for regular approval profiles if and only if it satisfies $\alpha$-AFS for fractional profiles. In more detail, if $f$ satisfies $\alpha$-AFS for fractional approval profiles if it satisfies this property also for regular aproval profiles because every regular approval profile corresponds to a fractional profile. On the other hand, if $f$ fails $\alpha$-AFS for a fractional profile $W$, there is a fractional subprofile $S$ and a candidate $x\in C$ such that $S(A)\leq W(A)$ for all $A\in 2^C\setminus \{\emptyset\}$, $S(A)>0$ implies $x\in A$, and $\frac{\alpha}{|S|} \cdot \sum_{A\in 2^C\setminus \{\emptyset\}} \left( S(A) \cdot \sum_{y\in A} f(W, y) \right) < \frac{|S|}{|W|}$. Since $f(\ell W)=f(W)$ for all $\ell\in\mathbb{N}$, it follows that $f$ also fails $\alpha$-AFS for every profile $\ell W$, which is witnessed by the fractional subprofile $\ell S$. However, because both $S$ and $W$ are fractional profiles and thus functions of the type $2^C\setminus \{\emptyset\}\rightarrow\mathbb{Q}$, there is an integer $\ell^*$ such that $\ell^*W(A)\in\mathbb{N}_0$ and $\ell^* S(A)\in\mathbb{N}_0$ for all approval ballots $A$. Hence, $\ell^*W$ corresponds to an approval profile and $\ell S^*$ to a group of voters in this approval profile, so $f$ fails $\alpha$-AFS for approval profiles, too. 
	
	We will next derive our upper bounds for $\frac{1}{\alpha}$ by investigating several fractional profiles. To this end, we fix some integer $\ell\geq 2$ and use $1+\ell(t-1)$ candidates $\{x^*\}\cup \{y_j^i\colon j\in [t-1], i\in [\ell]\}$. We then define the ballots $B^1,\dots, B^\ell$ by $B^i=\{x^*, y_1^i,\dots, y^i_{t-1}\}$ for all $i\in [\ell]$ and note that $|B^i|=t$ for all $i\in [\ell]$. 
	First, we consider the fractional profile $W^0$ defined by $W^0(B^i)=1$ for all $i\in [\ell]$ and $W^0(A)=0$ for all ballots $A\not\in \{B^1,\dots, B^\ell\}$. The total payment willingness for the candidate $x^*$ in the first round is $\ell \pi_1$ and the total payment willingness for all other candidates is $\pi_1$. Hence, $f$ chooses $x^*$ in the first round and assigns it a share of $\frac{\ell \pi_1}{|W^0|}=\pi_1$. 
	Moreover, since there is no overlap between the approval ballots $B^1,\dots, B^\ell$ except for $x^*$ and voters only spend their budget on their approved candidates, it holds for the distribution $p=f(W^0)$ that $\sum_{x\in B^i\setminus \{x^*\}} p(x)=\frac{1-\pi_1}{|W^0|}=\frac{1-\pi_1}{\ell}$ for all $i\in [\ell]$. Since every voter in $W^0$ approves $x^*$, $\alpha$-AFS requires that 
	\begin{align*}
		\sum_{A\in 2^C\setminus \{\emptyset\}} W^0(A) \sum_{x\in A} p(x)=\sum_{i=1}^\ell \sum_{x\in B^i} p(x)=\ell \pi_1+1-\pi_1\geq \frac{\ell}{\alpha}.
	\end{align*}
	
	Put differently, this means that $\pi_1 +\frac{1-\pi_1}{\ell}\geq  \frac{1}{\alpha}$. Now, if $\pi_1< \frac{1}{\alpha}$, there is an integer $\ell$ such that $\pi_1+\frac{1-\pi_1}{\ell}<\frac{1}{\alpha}$, which implies that $\alpha$-AFS is violated. Since $f$ satisfies this property by assumption, we derive that $\pi_1\geq  \frac{1}{\alpha}$. 
	
	Next, we consider $t-1$ more profiles $W^{z}_\epsilon$ with $z\in [t-1]$ and $\epsilon\in (0,1)$, where the rough idea is that each candidate in $\{y_j^i\colon j\in [z], i\in [\ell]\}$ is chosen before $x^*$ with a total payment willingness that is marginally higher than that of $x^*$. 
	To this end, we let $K(u,v)=(\ell-v)\cdot \pi_u + v\cdot \pi_{u+1}$ for $u\in [t-1]$ and $v\in [\ell]\cup \{0\}$ denote the payment willingness for $x^*$ when $v$ of the $\ell$ voters reporting $B^1,\dots, B^\ell$ have payed for $u$ candidates and the remaining $\ell-v$ voters have payed for $u-1$ candidates.
	We moreover note that $K(u,\ell)=\ell \pi_{u+1}=K(u+1,0)$ for all $u\in [t-1]$ and $K(u,v)\geq K(u,v+1)$ for all $u\in [t-1]$ and $v\in [\ell-1]\cup \{0\}$ since $\pi_1\geq \pi_2\geq \dots \geq \pi_t$. 
	Next, we define $K^\epsilon(u,v)$ for all $u\in [t-1]$, $v\in [\ell]\cup \{0\}$ as a function such that 
	\begin{enumerate}[label=\arabic*)]
		\item $K^\epsilon(u,v)\in \mathbb{Q}$,
		\item $K(u,v)< K^\epsilon(u,v)\leq K(u,v)+\epsilon$, and
		\item $K^\epsilon(u,v)> K^\epsilon(u',v')$ for all $u',v'$ with $u'>u$, or $u'=u$ and $v'>v$.
	\end{enumerate}
	
	Less formally, $K^\epsilon$ is an approximation of $K$ which takes care of two issues: firstly, $K(u,v)$ may be irrational while $K^\epsilon(u,v)$ is rational. Secondly, $K$ may assign the same score to multiple combinations of $u$ and $v$, whereas $K^\epsilon$ breaks these ties. We furthermore note that $K^\epsilon$ exists for every $\epsilon\in(0,1]$ due to the density of $\mathbb{Q}$ in $\mathbb{R}$.
	
	We now derive the profile $W^z_\epsilon$ from $W^0$ by adding $K^\epsilon(u,v-1)$ voters for each candidate $y_u^v$ with $u\in [z]$  and $v\in [\ell]$ that uniquely approve $y_u^v$. These voters ensure that $f$ processes the candidates in the order $y_1^1, \dots, y_1^\ell, y_2^1,\dots, y_2^\ell, \dots, y_z^1,\dots, y_z^\ell, x^*$ for the profile $W^z_\epsilon$. 
	To see this, fix two values $u\in [z]$ and $v\in [\ell]$ and assume that all candidates in $X=\{y_j^i\colon i\in [\ell], j\in [u-1]\}\cup \{y_u^i\colon i\in [v-1]\}$ have been assigned their shares. 
	In this case, the total payment willingness for $x^*$ is
	\begin{align*}
		\Pi(W^z_\epsilon, x^*,X) = \sum_{i\in[\ell]} \pi(t, |X\cap B^i|+1) = (\ell-v+1) \pi_{u} + (v-1) \pi_{u+1} = K(u,v-1).
	\end{align*}

	By contrast, the total payment willingness for a candidate $y_j^i\not \in X$ with $j\leq z$ is $\Pi(W^z_\epsilon, y_j^i,X) = K^\epsilon(j,i-1)+\pi(t, |B^i\cap X|+1)$, which is $K^\epsilon(j,i-1)+\pi_u$ if $i\geq v$ or $K^\epsilon(j,i-1)+\pi_{u+1}$ if $i\leq v-1$. Now, by the definition of $K^\epsilon$ and the fact that $\pi_u\geq \pi_{u+1}$, we have that $\Pi(W^z_\epsilon, y_u^v,X)>\Pi(W^z_\epsilon, y, X)$ for all candidates $y\not\in X\cup \{x^*, y_u^v\}$. Moreover, since $K^\epsilon(u,v-1)>K(u,v-1)$, we also have that $\Pi(W^z_\epsilon, y_u^v,X)>\Pi(W^z_\epsilon,x^*,X)$, which means that $y_u^v$ is the next candidate that will be assigned its share. Iteratively applying this reasoning results in the given sequence. 
	
	We next consider again the sets of voters $S$ that report $B^1, \dots, B^\ell$, i.e., $S$ is the subprofile of $W^z_\epsilon$ such that $S(B^i)=1$ for all $i\in [\ell]$ and $S(A)=0$ for all other ballots. In particular, we will next compute the total utility of the set of voters in $S$. To this end, let $p=f(W^z_\epsilon)$ be the distribution returned by $f$ for $W^z_\epsilon$. Since each candidate $y^i_j$ is approved by a single voter in $S$, it holds that
	\begin{align*} 
		\sum_{A\in 2^C\setminus \{\emptyset\}} S(A) \sum_{x\in A} p(x)=\ell \cdot p(x^*) + \sum_{j\in [t-1]} \sum_{i\in [\ell]}  p(y_j^i).
	\end{align*}
	
	First, we note that, by our previous analysis, it is easy to see that $x^*$ is chosen at a total payment willingness of $\Pi(W^z_\epsilon, x^*, \{y_j^i\colon j\in [z], i\in [\ell]\})=\ell \pi_{z+1}$. Next, every candidate $y_j^i$ with $j\leq z$ has a total payment willingness of $K^\epsilon(j,i-1)+\pi_j$ when it is assigned its share. Since $K^\epsilon(j,i-1)\leq K(j,i-1)+\epsilon$, we can compute that 
	\begin{align*}
		\sum_{j\in [z]} \sum_{i\in[\ell]} K^\epsilon(j,i-1) 
		&\leq \sum_{j\in [z]} \sum_{i\in[\ell]} K(j,i-1) +\epsilon\\
		&=z\cdot\ell\cdot \epsilon + \sum_{j\in [z]} \sum_{i\in[\ell]} \left((\ell-(i-1)) \pi_j +  (i-1) \pi_{j+1})\right)\\
		&=z\cdot\ell\cdot \epsilon + \sum_{i\in [\ell]} (\ell -(i-1)) \pi_1 + \sum_{i\in [\ell]} (i-1) \pi_{z+1} \\
		&\quad + \sum_{j=2}^{z} \sum_{i\in [\ell]} \left((\ell-(i-1)) \pi_j + (i-1)\pi_j\right)\\
		&=z\cdot\ell\cdot \epsilon + \pi_1 \frac{\ell(\ell+1)}{2} + \pi_{z+1} \frac{(\ell-1)\ell}{2} + \sum_{j=2}^z \ell^2 \pi_j
	\end{align*}
	
	Moreover, it holds that $\sum_{j\in [z]} \sum_{i\in[\ell]} \pi_j=\ell \sum_{j\in [z]} \pi_j$, and that $p(y_j^i)=\frac{\pi_{j+1}}{|W^z_\epsilon|}$ for all $j>z$ and $i\in [\ell]$ because these candidates are chosen after $y_{j'}^i$ for $j'\in [z]$ and $x^*$. Hence, we conclude that 
	\allowdisplaybreaks
	\begin{align*}
		\sum_{A\in 2^C\setminus \{\emptyset\}} S(A)\sum_{x\in A} p(x)&=\ell \cdot p(x^*) + \sum_{j\in [t-1]} \sum_{i\in [\ell]}  p(y_j^i) \\
		&\quad\leq \frac{1}{|W^z_\epsilon|}\bigg(\ell^2 \pi_{z+1} + \sum_{j\in [z]} \sum_{i\in [\ell]} (K^\epsilon(j,i-1)+\pi_j) +  \sum_{j=z+2}^{t} \sum_{i\in [\ell]} \pi_j\bigg)\\
		&\quad\leq\frac{1}{|W^z_\epsilon|} \bigg(\ell^2 \pi_{z+1} + z\cdot\ell\cdot \epsilon + \pi_1 \frac{\ell(\ell+1)}{2} + \pi_{z+1} \frac{(\ell-1)\ell}{2} \\
		&\quad\quad+ \sum_{j=2}^z \ell^2 \pi_j + \sum_{j=1}^z \ell \pi_j + \ell (1-\sum_{j=1}^{z+1} \pi_j) \bigg)\\
		&\quad\leq \frac{\ell^2}{|W^z_\epsilon|}\bigg(\pi_{z+1} + \frac{z\epsilon}{\ell} + 	\frac{\pi_1}{2} + \frac{\pi_1}{2\ell} +\frac{\pi_{z+1}}{2} + 
		 \sum_{j=2}^z \pi_j +  \frac{1}{\ell}\bigg).\\
		&\quad= \frac{\ell^2}{|W^z_\epsilon|}\bigg(\frac{\pi_1}{2} + \frac{\pi_1}{2\ell} + \sum_{j=2}^z \pi_j + \frac{3\pi_{z+1}}{2} + \frac{1}{\ell}+\frac{z\epsilon}{\ell}\bigg).
	\end{align*}
	
	Furthermore, $\alpha$-AFS requires that $\sum_{A\in 2^C\setminus \{\emptyset\}} S(A) \sum_{x\in A} p(x)\geq \frac{1}{\alpha} \frac{\ell^2}{|W^z_\epsilon|}$. Hence, we conclude that 
	\begin{align*}
		\frac{\pi_1}{2} + \frac{\pi_1}{2\ell} + \sum_{j=2}^z \pi_j + \frac{3\pi_{z+1}}{2} + \frac{1}{\ell}+\frac{z\epsilon}{\ell}\geq \frac{1}{\alpha}.
	\end{align*}
	
	Now, if $\frac{\pi_1}{2} + \sum_{j=2}^z \pi_j + \frac{3\pi_{z+1}}{2} <\frac{1}{\alpha}$, we can find a sufficiently large $\ell$ (and sufficiently small $\epsilon$) such that the above inequality is violated. This contradicts $\alpha$-AFS, so $\frac{\pi_1}{2} + \sum_{j=2}^z \pi_j + \frac{3\pi_{z+1}}{2} \geq \frac{1}{\alpha}$. Since this holds for every $z\in [t-1]$, the proposition follows. 
\end{proof}

We note that \Cref{prop:lowerbound} is a powerful tool to infer lower bounds on the AFS approximation ratios of sequential payment rules. For instance, if $\pi(t,1)=\frac{2}{3}$ for some $t\in\mathbb{N}$, the first inequality of this proposition implies that the corresponding sequential payment rule only satisfies $\alpha$-AFS for $\alpha\geq\frac{3}{2}$.
Moreover, we will next use this proposition to fully determine the AFS approximation ratio of all $\gamma$-multiplicative sequential payment (\mps{\gamma}) rules.
To this end, we recall that, for every $\gamma\in [0,1]$, \mps{\gamma} is the sequential payment rule defined by the payment willingness function $\pi$ that satisfies 
\emph{(i)} $\sum_{y\in [x]} \pi(x,y)=1$ for all $x\in\mathbb{N}$ and 
\emph{(ii)} $\pi(x,y+1)=\gamma\cdot \pi(x,y)$ for all $x\in\mathbb{N}$ and $y\in [x-1]$. 
When $\gamma>0$, this means that $\pi(x,y)=\frac{\gamma^y}{\sum_{i\in [x]} \gamma^i}$ for all $x\in\mathbb{N}$ and $y\in [x]$. On the other hand, when $\gamma=0$, then $\pi(x,1)=1$ for all $x\in\mathbb{N}$, so \mps{0} is equivalent to \map.
We note that the following proposition formalizes the claims made in \Cref{rem:SPR}. %Moreover, it will turn out helpful to prove the upper bound for \Cref{thm:tightbound}.

\begin{restatable}{proposition}{MPS}
	\label{thm:MPS}
	\mps{\gamma} satisfies $\frac{2}{1+\gamma}$-AFS if $0\leq \gamma\leq \frac{1}{3}$ and $\frac{1}{1-\gamma}$-AFS if $\frac{1}{3}\leq \gamma<1$. Moreover, these bounds are tight.
\end{restatable}
\begin{proof}
	Since the $0$-multiplicative sequential spending rule is the maximum payment rule, the theorem follows for $\gamma=0$ from \Cref{thm:MPmain}. Hence, we assume that $\gamma\in (0,1)$ and we will subsequently prove an upper and lower bound on the approximation ratio of \mps{\gamma}.
	As usual, we will denote by $\pi$ the payment willingness function of \mps{\gamma}, i.e., $\pi(x,y)=\frac{\gamma^y}{\sum_{i=1}^x \gamma^i}$ for all $x\in \mathbb{N}$, $y\in [x]$. We start by showing the lower bound on the AFS approximation ratio of $f$.
	\medskip

	\textbf{Claim 1: Lower Bound}
	
	To prove our lower bound, we use a case distinction with respect to $\gamma$. First, consider the case that $\frac{1}{3}\leq \gamma<1$. In this case, we use the first inequality of \Cref{prop:lowerbound}, which shows that, if $f$ satisfies $\alpha$-AFS for some $\alpha\in\mathbb{R}$, then $\pi(t,1)\geq \frac{1}{\alpha}$ for all $t\in\mathbb{N}$. Since $\pi(t,1)=\frac{\gamma}{\sum_{i\in [t]}\gamma^i}$, this means that, for all $t\in\mathbb{N}$, it holds that 
	$\alpha\geq \sum_{i=0}^{t-1} \gamma^i$.
	The right hand side of this inequality is the geometric series which converges to $\sum_{i\in\mathbb{N}_0} \gamma^i =\frac{1}{1-\gamma}$, so it follows that $\alpha$ cannot be smaller than $\frac{1}{1-\gamma}$. In more detail, for every $\epsilon>0$, there is some $t\in\mathbb{N}$ such that $\frac{1}{1-\gamma}-\epsilon< \sum_{i=0}^{t-1} \gamma^i$ and \Cref{prop:lowerbound} would thus be violated if \mps{\gamma} satisfies $\frac{1}{1-\gamma}-\epsilon$-AFS.
	
	As the second case, suppose that $0<\gamma<\frac{1}{3}$. We consider the last inequality given by \Cref{prop:lowerbound}, which shows that, for all $t\in\mathbb{N}$, $\frac{1}{2} \pi(t,1) + \frac{3}{2} \pi(t,t) + \sum_{i=2}^{t-1} \pi(t,i)\geq \frac{1}{\alpha}$ if $f$ satisfies $\alpha$-AFS. Since $\sum_{i\in [t]} \pi(t,i)=1$, we can rewrite this inequality as 
		$1-\frac{1}{2}\pi(t,1)+\frac{1}{2}\pi(t,t)\geq \frac{1}{\alpha}$ or, equivalently,
		$\alpha\geq \frac{2}{2-\pi(t,1)+\pi(t,t)}$.	
	Since $\pi(t,1)=\frac{\gamma}{\sum_{i\in [t]}\gamma^i}\geq \frac{\gamma}{\sum_{i\in \mathbb{N}}\gamma^i}=1-\gamma$, this means that $\alpha\geq \frac{2}{1+\gamma+\pi(t,t)}$. 
	Finally, we note that $\pi(t,t)=\frac{\gamma^t}{\sum_{i\in [t]} \gamma^i}\leq \gamma^{t-1}$, which converges to $0$ as $t$ goes to infinity. Hence, for every $\epsilon>0$, there is a $t\in\mathbb{N}$ such that $\frac{2}{1+\gamma}-\epsilon<\frac{2}{1+\gamma+\pi(t,t)}$, so $f$ fails $\alpha$-AFS for every $\alpha<\frac{2}{1+\gamma}$.\medskip		 
	
	\textbf{Claim 2: Upper Bound}
	
	To show our upper bound on the AFS approximation ratio of $f$, let $\mathcal{A}$ denote an approval profile, let $S\subseteq N$ be a group of voters, $x^*$ a candidate such that $x^*\in \bigcap_{i\in S} A_i$, and $p$ the distribution chosen by \mps{\gamma} for $\mathcal A$. We moreover define $s=|S|$ and $n=|N|$ for a simple notation, and let $t=\max_{i\in N} |A_i|$ denote the maximal ballot size in $\mathcal{A}$. Furthermore, let $x_1,\dots, x_m$ denote the sequence in which the \mps{\gamma} rule processes the candidates for $\mathcal{A}$ and let $k$ denote the index such that $x_k=x^*$. It holds that
	\begin{align*}®
		\sum_{i\in S} u_i(p)=\frac{1}{n}\sum_{j\in [m]} |N_{x_j}(\mathcal{A})\cap S|\cdot \Pi(\mathcal{A}, x_j, \{x_1,\dots, x_{j-1}\}). 
	\end{align*}
	
	By the definition of sequential payment rules, we have for every candidate $x_j$ with $j\leq k$ that 
	\begin{align*}
		\Pi(\mathcal{A}, x_j, \{x_1,\dots, x_{j-1}\})&\geq \Pi(\mathcal{A}, x^*,  \{x_1,\dots, x_{j-1}\})\\
		&\geq \sum_{i\in S} \pi(|A_i|, |A_i\cap \{x_1,\dots, x_{j-1}\}|+1). 
	\end{align*}
	
	Moreover, it holds that $\pi(|A_i|, |A_i\cap \{x_1,\dots, x_{j-1}\}|+1)\geq \pi(t,  |A_i\cap \{x_1,\dots, x_{j-1}\}|+1)$ because $|A_i|\leq t$ and 
		$\pi(x,y)=\frac{\gamma^{y}}{\sum_{i=1}^{x} \gamma^{i}}> \frac{\gamma^{y}}{\sum_{i=1}^{x+1} \gamma^{i}}=\pi(x+1,y)$
	for all $x<t$, $y\in [x]$. Consequently, we conclude that 
	\begin{align*}
		\sum_{i\in S} &u_i(p)\geq \frac{1}{n} \sum_{j\in [k]} \left(|N_{x_j}(\mathcal{A})\cap S| \cdot \sum_{i\in S} \pi(t, |A_i\cap \{x_1,\dots, x_{j-1}\}|+1)\right). 
	\end{align*}

	For the next step, we define $\pi_i=\pi(t,i)$ for all $i\in [t]$ and $T_j=\sum_{\ell=1}^{j-1} |N_{x_{j-1}}(\mathcal{A})\cap S|$ as the total number of approvals that the voters in $S$ give to the candidates in $\{x_1,\dots, x_{j-1}\}$. We will now derive a lower bound for $\sum_{i\in S} \pi(t, |A_i\cap \{x_1,\dots, x_{j-1}\}|+1)$ depending on $T_j$. To this end, let $w^j$ denote the vector defined by $w^j_\ell=|\{i\in S\colon |A_i\cap \{x_1,\dots,x_{j-1}\}|+1=\ell\}|$ for all $\ell\in [t]$, i.e., $w^j_\ell$ states the number of voters in $S$ that approve $\ell-1$ of the candidates in $\{x_1,\dots, x_{j-1}\}$. It holds by definition that $\sum_{\ell\in [t]} w^j_\ell\cdot (\ell-1)=T_j$ and $\sum_{\ell\in [t]} w^j_\ell=s$. 
	Moreover, we have that \[\sum_{i\in S} \pi(t, |A_i\cap \{x_1,\dots, x_{j-1}\}|+1) = \sum_{\ell\in [t]} w^j_\ell \pi_{\ell}.\] 
	Now, assume that there are two indices $\ell_1$ and $\ell_2$ such that $\ell_2-\ell_1\geq 2$, $w^j_{\ell_1}>0$, and $w^j_{\ell_2}>0$. We consider next the vector $\hat w$ with $\hat w_{z}=w^j_z$ for all $z\in [t]\setminus \{\ell_1, \ell_1+1, \ell_2-1, \ell_2\}$, $\hat w_{\ell_1}=w^j_{\ell_1}-1$, $\hat w_{\ell_1+1}=w^j_{\ell_1+1}+1$, $\hat w_{\ell_2-1}=w^j_{\ell_2-1}+1$, and $\hat w_{\ell_2}=w^j_{\ell_2}-1$. It holds that
	\begin{align*}
		\sum_{\ell\in [t]} w^j_\ell \pi_{\ell}-\sum_{\ell\in [t]} \hat w_\ell \pi_{\ell}&= (\pi_{\ell_1}-\pi_{\ell_1+1})+(\pi_{\ell_2}-\pi_{\ell_2-1})\\
		&=(1-\gamma) \pi_{\ell_1} -(1-\gamma)\pi_{\ell_2-1}\\
		&>0. 
	\end{align*}
	
	The second equation uses that $\gamma \pi_\ell=\pi_{\ell+1}$ for all $\ell\in [m-1]$ and the final inequality that $\pi_{\ell_1}>\pi_{\ell_2-1}$ (which holds since $\ell_1<\ell_2-1$). Moreover, our construction ensures that $\sum_{\ell\in [t]} \hat w_\ell\cdot (\ell-1)=T_j$ and $\sum_{\ell\in [t]} \hat w_\ell=s$. Now, we can repeat this process until we arrive at a vector $w^*$ such that there are no indices $\ell_1$ and $\ell_2$ with $\ell_2-\ell_1\geq 2$, $ w^*_{\ell_1}>0$, and $w^*_{\ell_2}>0$. In particular, for this vector, it still holds that $\sum_{\ell\in [t]} w^*_\ell\cdot (\ell-1)=T_j$ and $\sum_{\ell\in [t]} w^*_\ell=s$ and that $\sum_{\ell\in [t]} w^j_\ell \pi_{\ell}>\sum_{\ell\in [t]} w^*_\ell \pi_{\ell}$. Due to our constraints, we can fully specify $w^*$. To this end, let $\ell^*=\lfloor \frac{T_j}{s}\rfloor+1$. Since $w^*$ can only have two non-zero entries and $\sum_{\ell\in [t]} w^*_\ell\cdot (\ell-1)=T_j$ and $\sum_{\ell\in [t]} w^*_\ell=s$, it must be that  $w^*_{\ell^*}=s - (T_j \mod s)$, $w^*_{\ell^*+1}=T_j\mod s$, and $w^*_\ell=0$ for all $\ell\in [t]\setminus \{\ell^*, \ell^*+1\}$. Finally, we infer that
	\begin{align*}
		\sum_{i\in S} \pi(t, |A_i\cap \{x_1,\dots, x_{j-1}\}|+1)&= \sum_{\ell\in [t]} w^j_\ell \pi_{\ell}\\
		&\geq \sum_{\ell\in [t]} w^*_\ell \pi_{\ell}\\
		&=(s-(T_j\mod s)) \pi_{\ell^*} + (T_j\mod s) \pi_{\ell^*+1}. 
	\end{align*}
	
	We next define $K(u,v)=(s-v) \pi_u+ v\pi_{u+1}$ for all $u\in [{t-1}]$ and $v\in [s]\cup \{0\}$ and note that $K(\lfloor\frac{T_j}{s}\rfloor +1, T_j\mod s)=\sum_{\ell\in [t]} w^*_\ell \pi_{\ell}$. 
	By our previous analysis, it holds that $\sum_{i\in S} \pi(t, |A_i\cap \{x_1,\dots, x_{j-1}\}|+1)\geq K(\lfloor\frac{T_j}{s}\rfloor +1, T_j\mod s)$ for all $j\in [k]$ and thus, 
	\begin{align*}
		\sum_{i\in S} u_i(p)\geq \frac{1}{n}\sum_{j\in [k]} |N_{x_j}(\mathcal{A})\cap S| \cdot K(\lfloor\frac{T_j}{s}\rfloor +1, T_j\text{ mod } s). 
	\end{align*}
	
	By the definition of $K$, we have that $K(u,v)>K(u,v+1)$ for all $u\in [t-1]$ and $v\in [s-1]\cup \{0\}$ and $K(u,s)=K(u+1,0)$. This implies that
	\begin{align*}
		|N_{x_j}(\mathcal{A})\cap S|\cdot K(\lfloor\frac{T_j}{s}\rfloor +1, T_j\text{ mod } s)
		&= \sum_{T=T_j}^{T_j+|N_{x_j}(\mathcal{A})\cap S|-1} K(\lfloor\frac{T_j}{s}\rfloor +1, T_j\text{ mod } s)\\
		&\geq \sum_{T=T_j}^{T_j+|N_{x_j}(\mathcal{A})\cap S|-1} K(\lfloor\frac{T}{s}\rfloor +1, T\text{ mod } s). 
	\end{align*}
	
	Moreover, it holds hat $|N_{x_k}(\mathcal{A})\cap S|=s$ as all voters in $S$ approve $x_k=x^*$. This means that 
	\begin{align*}
		\sum_{i\in S} u_i(p)\geq \frac{1}{n}\sum_{T=0}^{T_k-1} K(\lfloor\frac{T}{s}\rfloor +1, T\text{ mod } s) 
		+ \frac{s}{n}\cdot K(\lfloor\frac{T_k}{s}\rfloor +1, T_k\text{ mod } s).
	\end{align*}
	
	Next, we define $u=\lfloor\frac{T_k}{s}\rfloor +1$ and $v=T_k\mod s$. Based on this notation, we infer that 
	\begin{align*}
		\sum_{T=0}^{T_k-1} K(\lfloor\frac{T}{s}\rfloor +1, T\text{ mod } s)&= \sum_{j=1}^{u-1} \sum_{\ell=0}^{s-1} K(j,\ell) + \sum_{\ell=0}^{v-1} K(u,\ell). 
	\end{align*} 
	
	We recall that $\gamma \pi_j=\pi_{j+1}$, so we calculate that 
	\begin{align*}
		\sum_{\ell=0}^{s-1} K(j,\ell) &= \sum_{\ell=0}^{s-1} \left((s-\ell) \cdot \pi_j + \ell\cdot \pi_{j+1}\right)\\
		& = \sum_{\ell=0}^{s-1} (s-(1-\gamma)\ell) \cdot \pi_j \\
		&=s^2\cdot \pi_j - (1-\gamma) \cdot \frac{(s-1)s}{2}\cdot \pi_j\\
		&\geq \frac{1+\gamma}{2}\cdot s^2\cdot \pi_j.
	\end{align*}
	
	Next, we observe that $\pi_j=\gamma^{j-1}\pi_1$ by the definition of \mps{\gamma}. Using the geometric sum, we derive that
	\begin{align*}
		\sum_{j=1}^{u-1} \sum_{\ell=0}^{s-1} K(j,\ell)& \geq\sum_{j=1}^{u-1} \left(\frac{1+\gamma}{2}\cdot s^2\cdot \pi_j\right)\\
		&= \frac{1+\gamma}{2}\cdot s^2\cdot \pi_1 \cdot\sum_{j=0}^{u-2} \gamma^j\\
		&=  s^2\cdot \pi_1 \cdot \frac{1+\gamma}{2}\cdot \frac{1-\gamma^{u-1}}{1-\gamma}.
	\end{align*}
	
	We will next use a case distinction with respect to $\gamma$ to complete the proof. \medskip
	
	\emph{Case 1: $\gamma\geq \frac{1}{3}$.} In this case, we compute the following lower bound for $\sum_{\ell=0}^{v-1} K(u,\ell)+s\cdot K(u,v)$:
	\begin{align*}
		\sum_{\ell=0}^{v-1} K(u,\ell)+s\cdot K(u,v)
		&= \sum_{\ell=0}^{v-1} \left( (s-\ell)\cdot \pi_u + \ell \cdot \pi_{u+1}\right) +s\cdot((s-v) \cdot \pi_u + v\cdot \pi_{u+1})\\
		&=\sum_{\ell=0}^{v-1} \left((s-(1-\gamma)\ell)\cdot \pi_u \right) +s\cdot ((s-(1-\gamma)v) \cdot \pi_u)\\
		&=\pi_u\cdot \big(v\cdot s -(1-\gamma)\cdot \frac{(v-1)v}{2} + s^2 - (1-\gamma)v\cdot s\big)\\
		&\geq \pi_u \cdot \big(s^2+v\cdot s\cdot (1-\frac{3}{2}(1-\gamma))\big)\\
		&\geq\pi_u\cdot s^2\\
		&=s^2\cdot \gamma^{u-1}\cdot \pi_1.
	\end{align*}
	
	In our second to last line, we use that $\gamma\geq \frac{1}{3}$, so $1-\frac{3}{2}(1-\gamma)\geq 0$. Based on our observations, we now conclude that 
	
	\begin{align*}
		\sum_{i\in S} u_i(p)&\geq \frac{1}{n}\sum_{T=0}^{T_k-1} K(\lfloor\frac{T}{s}\rfloor +1, T\text{ mod } s) + \frac{s}{n}\cdot K(\lfloor\frac{T_k}{s}\rfloor +1, T_k\text{ mod } s)\\
		&=\frac{1}{n}\bigg(\sum_{j=1}^{u-1} \sum_{\ell=0}^{s-1} K(j,\ell) + \sum_{\ell=0}^{v-1} K(u,\ell) + s\cdot K(u,v)\bigg)\\
		&\geq\frac{s^2}{n}\pi_1\bigg(\frac{1+\gamma}{2}\cdot \frac{1-\gamma^{u-1}}{1-\gamma} + \gamma^{u-1}\bigg)\\
		&=\frac{s^2}{n}\pi_1\cdot\frac{1+\gamma-\gamma^{u-1}-\gamma^u+2\gamma^{u-1} - 2\gamma^u}{2(1-\gamma)}\\
		&=\frac{s^2}{n}\pi_1\cdot \frac{1+\gamma +\gamma^{u-1} (1-3\gamma)}{2(1-\gamma)}\\
		&\geq \frac{s^2}{n}\pi_1\cdot \frac{1+\gamma  +1-3\gamma}{2(1-\gamma)}\\
		&=\frac{s^2}{n}\pi_1.
	\end{align*}
	
	The first five inequalities here use our previous bounds and some standard transformation. For the sixth line (where we replace $\gamma^{u-1}$ with $1$), we use the fact that $1-3\gamma\leq 0$ and $\gamma^{u-1}\leq 1$ as $\frac{1}{3}\leq \gamma<1$ and $u\geq 1$. Finally, employing the geometric series, we infer that 
	\begin{align*}
		\pi_1=\frac{\gamma}{\sum_{i=1}^t \gamma^i}\geq \frac{\gamma}{\sum_{i\in\mathbb N} \gamma^i}=\frac{1}{\sum_{i\in \mathbb{N}_0} \gamma^i}=\frac{1}{1/(1-\gamma)}=1-\gamma.
	\end{align*}
	
	By combining this insight with the fact that $\sum_{i\in S} u_i(p)\geq \frac{s^2}{n} \pi_1$, it finally follows that \mps{\gamma} satisfies $\frac{1}{1-\gamma}$-AFS if $\gamma\geq\frac{1}{3}$.\medskip
	
	\emph{Case 2: $\gamma<\frac{1}{3}$.} Next, assume that $\gamma<\frac{1}{3}$. In this case, it holds that
	\begin{align*}
		\sum_{\ell=0}^{v-1} K(u,\ell)+s\cdot K(u,v)
		%&= \sum_{\ell=0}^{v-1} (s-\ell)\cdot \pi_u + \ell \cdot \pi_{u+1} +s\cdot((s-v) \cdot \pi_u + v\cdot \pi_{u+1})\\
		%&=\sum_{\ell=0}^{v-1} (s-(1-\gamma)\ell)\cdot \pi_u +s\cdot ((s-(1-\gamma)v) \cdot \pi_u)\\
		%&=\pi_u\cdot \big(v\cdot s -(1-\gamma)\cdot \frac{(v-1)v}{2} + s^2 - (1-\gamma)v\cdot s\big)\\
		&\geq \pi_u \cdot \big(s^2+v\cdot s\cdot (1-\frac{3}{2}(1-\gamma))\big)\\
		&\geq\pi_u \cdot \big(s^2+s^2\cdot (1-\frac{3}{2}(1-\gamma))\big)\\
		&=\pi_u \cdot s^2\cdot (\frac{1}{2}+\frac{3}{2}\gamma)\\
		&=\gamma^{u-1}\cdot \pi_1 \cdot s^2\cdot (\frac{1}{2}+\frac{3}{2}\gamma)\\
	\end{align*}
	
	Note here that the first inequality can be derived analogously to the last case. However, we now have that  $(1-\frac{3}{2}(1-\gamma))<0$, which means that $v\cdot s\cdot (1-\frac{3}{2}(1-\gamma))\geq s^2\cdot (1-\frac{3}{2}(1-\gamma))$. In a similar fashion to the first case, we can now derive that 
	\begin{align*}
		\sum_{i\in S} u_i(p)
		&\geq\frac{s^2}{n}\pi_1\bigg(\frac{1+\gamma}{2}\cdot \frac{1-\gamma^{u-1}}{1-\gamma} + \gamma^{u-1}\cdot (\frac{1}{2}+\frac{3}{2}\gamma)\bigg)\\
		&=\frac{s^2}{n}\pi_1 \cdot \frac{1+\gamma -\gamma^{u-1}-\gamma^u +\gamma^{u-1}+3\gamma^u-\gamma^u-3\gamma^{u+1}}{2(1-\gamma)}\\
		&=\frac{s^2}{n}\pi_1 \cdot \frac{1+\gamma +\gamma^u(1-3\gamma)}{2(1-\gamma)}\\
		&\geq \frac{s^2}{n}\pi_1 \cdot \frac{1+\gamma}{2(1-\gamma)}.
	\end{align*}
	
	We use here for the last inequality that $1-3\gamma>0$. Finally, by using again that $\pi_1\geq {1-\gamma}$, it follows that $\sum_{i\in S} u_i(p)\geq \frac{s^2}{n}\cdot \frac{1+\gamma}{2}$, which proves that \mps{\gamma} satisfies $\frac{2}{1+\gamma}$-AFS if $0<\gamma<\frac{1}{3}$. 
\end{proof}

Finally, we will prove \Cref{thm:SPR}.

\tightanalysis*
\begin{proof}
	We split the theorem into three statements. Firstly, we will show that no sequential payment rule satisfies $\alpha$-AFS for $\alpha< \frac{3}{2}(1-3^{-t})$ when $t$ denotes the maximum feasible ballot size. Secondly, we will prove that each sequential payment rule except for the $\frac{1}{3}$-MSP rule also fails $\alpha$-AFS for $\alpha=\frac{3}{2}(1-3^{-t})$ for some maximal feasible ballot size $t$. As the last point, we will show that the $\frac{1}{3}$-multiplicative sequential payment rule indeed matches this lower bound.\medskip
	
	\textbf{Claim 1: No sequential payment rule satisfies $\alpha$-AFS for $\alpha< \frac{3}{2}(1-3^{-t})$ when the maximal feasible ballot size is $t$.}
	
	Let $f$ denote a sequential payment rule that satisfies $\alpha$-AFS for some $\alpha\in\mathbb{R}$, let $\pi$ be its payment willingness function, and let $t\in\mathbb{N}$ denote the maximum feasible ballot size. We moreover define $\pi_i=\pi(t,i)$ for all $i\in [t]$ and note that $\pi_1\geq \pi_2\geq \dots \geq \pi_t$ by definition. Using \Cref{prop:lowerbound}, we infer the following inequalities:
	\begin{align*}
		\pi_1\geq  \frac{1}{\alpha}\\
		\frac{1}{2}\pi_1 + \frac{3}{2} \pi_2\geq \frac{1}{\alpha}\\
		\frac{1}{2}\pi_1 + \pi_2 + \frac{3}{2} \pi_3\geq  \frac{1}{\alpha}\\
		\dots\\
		\frac{1}{2}\pi_1 + \pi_2 + \dots + \pi_{t-1}+\frac{3}{2} \pi_t\geq  \frac{1}{\alpha}.
	\end{align*}
	
	Based on these inequalities, we will now derive an upper bound on $\frac{1}{\alpha}$ by using the fact that $\pi_1+\pi_2+\dots+\pi_t=1$. To this end, we first note that by multiplying the $i$-th inequality (for $i\geq 2$) with $2$ and adding the first inequality, we obtain the following modified set of inequalities. 
	\begin{align*}
		\pi_1\geq  \frac{1}{\alpha}\\
		2\pi_1 + 3 \pi_2\geq \frac{3}{\alpha}\\
		2\pi_1 + 2 \pi_2 + 3 \pi_3\geq \frac{3}{\alpha}\\
		\dots\\
		2\pi_1 + 2\pi_2 + \dots + 2\pi_{t-1}+3 \pi_t\geq \frac{3}{\alpha}.
	\end{align*}
	
	Now, by adding up the first and second constraint, we get that $3\pi_1+3\pi_2\geq \frac{4}{\alpha}$. Using this constraint and adding in three times the third constraint, we derive that $9\pi_1+9\pi_2+9\pi_3\geq \frac{13}{\alpha}$. As next step, we can add $9$ times our fourth constraint to get that $27\pi_1+27\pi_2+27\pi_3+27\pi_4\geq \frac{40}{\alpha}$. More generally, denoting the left hand sides of our modified inequalities by $I^1,\dots, I^t$, it holds that $I^1+I^2+3I^3+9I^4+\dots+3^{t-2}I^t=3^{t-1}(\pi_1+\pi_2+\dots \pi_t)=3^{t-1}$. Moreover, when adding up the right hand sides with the same coefficients, we get that $\frac{1}{\alpha}+\frac{3}{\alpha}+\frac{9}{\alpha}+\dots+\frac{3^{t-1}}{\alpha}=\frac{1}{\alpha} \sum_{i=0}^{t-1} 3^i$. Finally, solving for $\alpha$ shows that $\alpha\geq \frac{\sum_{i=0}^{t-1} 3^i}{3^{t-1}}= \frac{3}{2} (1-3^{-t})$ as $\sum_{i=0}^{t-1} 3^i=\frac{1}{2}(3^t-1)$, which proves our lower bound. \bigskip
	
	\textbf{Claim 2: No sequential payment rule other than the $\frac{1}{3}$-MSP rule satisfies $\frac{3}{2}(1-3^{-t})$-AFS for all maximum feasible ballot sizes $t\in\mathbb{N}$.}
	
	Let $f$ denote a sequential payment rule, let $\pi$ denote its payment willingness function, fix some arbitrary integer $t\in\mathbb{N}$ with $t\geq 2$, and define $\pi_i$ as in Claim 1. We will show that $\pi_i=\frac{3^{-i}}{\sum_{j\in [t]} 3^{-j}} $ for all $i\in [t]$ if $f$ satisfies $\frac{3}{2}(1-3^{-t})$-AFS for the maximal ballot size $t$. Because this holds for all $t\in\mathbb{N}$, this proves that $f$ is \mps{\frac{1}{3}}, so no other rule satisfies our fairness condition. Now, to show this claim, we use again the inequalities of \Cref{prop:lowerbound}. In particular, we note here that $\frac{3}{2}(1-3^{-t})=\sum_{i=0}^{t-1} 3^{-i}$. Hence, \Cref{prop:lowerbound} shows that
	\begin{align*}
		\pi_1\geq  \frac{1}{\sum_{i=0}^{t-1} 3^{-i}}\\
		\frac{1}{2}\pi_1 + \frac{3}{2} \pi_2\geq \frac{1}{\sum_{i=0}^{t-1} 3^{-i}}\\
		\frac{1}{2}\pi_1 + \pi_2 + \frac{3}{2} \pi_3\geq  \frac{1}{\sum_{i=0}^{t-1} 3^{-i}}\\
		\dots\\
		\frac{1}{2}\pi_1 + \pi_2 + \dots + \pi_{t-1}+\frac{3}{2} \pi_t\geq  \frac{1}{\sum_{i=0}^{t-1} 3^{-i}}\\
	\end{align*}
	
	Our first inequality immediately implies that $\pi_1\geq \frac{1}{\sum_{i=0}^{t-1} 3^{-i}}=\frac{3^{-1}}{\sum_{i\in [t]} 3^{-i}}$. Next, since $\sum_{i\in [t]} \pi_i=1$, we can rewrite our last inequality by 
	\begin{align*}
		\frac{1}{2}\pi_1 + \pi_2 + \dots + \pi_{t-1}+\frac{3}{2} \pi_t&\geq \frac{1}{\sum_{i=0}^{t-1} 3^{-i}}\\
		\iff 1-\frac{1}{2} \pi_1 + \frac{1}{2}\pi_t&\geq \frac{1}{\sum_{i=0}^{t-1} 3^{-i}}\\
		\iff \pi_t&\geq \frac{2}{\sum_{i=0}^{t-1} 3^{-i}} + \pi_1-2.
	\end{align*}
	
	Using the fact that $\pi_1\geq\frac{1}{\sum_{i=0}^{t-1} 3^{-i}}$ and that $\sum_{i=0}^{t-1} 3^{-i}=\frac{3}{2}(1-3^{-t})$, we get that 
	\begin{align*}
		\pi_t\geq \frac{3}{\sum_{i=0}^{t-1} 3^{-i}} -\frac{2\sum_{i=0}^{t-1} 3^{-i}}{\sum_{i=0}^{t-1} 3^{-i}}
		\geq \frac{3^{-(t-1)}}{\sum_{i=0}^{t-1} 3^{-i}}=\frac{3^{-t}}{\sum_{i=1}^t 3^{-i}}.
	\end{align*}
	
	Next, assume inductively that there is $j\in \{3,\dots,t\}$ such that $\pi_{\ell}\geq \frac{3^{-(\ell-1)}}{\sum_{i=0}^{t-1} 3^{-i}}$ for all $\ell\in \{j,\dots, t\}$. We will show that the same holds for $j-1$. To this end, we first note that 
	\begin{align*}
		\frac{1}{2}\pi_1 + \pi_2 + \dots + \pi_{j-2} + \frac{3}{2} \pi_{j-1} & \geq  \frac{1}{\sum_{i=0}^{t-1} 3^{-i}}\\
		\iff (1-\sum_{\ell=j}^t \pi_\ell) + \frac{1}{2} \pi_{j-1} -  \frac{1}{2} \pi_{1} & \geq   \frac{1}{\sum_{i=0}^{t-1} 3^{-i}}\\
		\iff \pi_{j-1}&\geq  \frac{2}{\sum_{i=0}^{t-1} 3^{-i}} + \pi_{1} + 2\sum_{\ell=j}^t \pi_\ell -2.
	\end{align*}
	
	Using the same computations as $\pi_t$ shows that $ \frac{2}{\sum_{i=0}^{t-1} 3^{-i}} + \pi_{1} -2\geq \frac{3^{-(t-1)}}{\sum_{i=0}^{t-1} 3^{-i}}$. Combining this with the lower bounds on $\pi_\ell$ for $\ell\in \{j,\dots, t\}$ given by the induction hypothesis, we conclude that 
	\begin{align*}
		\pi_{j-1}&\geq \frac{3^{-(t-1)}}{\sum_{i=0}^{t-1} 3^{-i}} + 2\sum_{\ell=j}^{t} \frac{3^{-(\ell-1)}}{\sum_{i=0}^{t-1} 3^{-i}}\\
		&=\frac{3^{-(j-1)}}{\sum_{i=0}^{t-1} 3^{-i}} \left(3^{-(t-j)} + 2\sum_{\ell=0}^{t-j} 3^{-\ell}\right)\\
		&=\frac{3^{-(j-1)}}{\sum_{i=0}^{t-1} 3^{-i}} \left(3^{-(t-j)} + 3-3^{-(t-j)}\right)\\
		&=\frac{3^{-(j-2)}}{\sum_{i=0}^{t-1} 3^{-i}}.
	\end{align*}
	
	Hence, it holds that $\pi_j\geq \frac{3^{-(j-1)}}{\sum_{i=0}^{t-1} 3^{-i}}=\frac{3^{-j}}{\sum_{i\in [t]} 3^{-i}}$ for all $j\in [t]$. Since $\sum_{j\in [t]} \pi_j=1$, all of these inequalities must be tight, which means that the payment willingness function $\pi$ for approval ballots of size $t$ is the equal to the payment willingness of \mps{\frac{1}{3}}. Because this analysis holds for all $t$, it follows that \mps{\frac{1}{3}} is the only sequential payment rule that can satisfy our lower bound tightly for all maximum ballot sizes.\bigskip
	
	\textbf{Claim 3: \mps{\frac{1}{3}} satisfies $\frac{3}{2}(1-3^{-t})$-AFS for all profiles $\mathcal{A}$ with maximum ballot size $t\in\mathbb{N}$.}
	
	Let $\mathcal{A}$ denote a profile such that $|A_i|\leq t$ for all $i\in N$, and let $S\subseteq N$ be a group of voters and $x^*$ a candidate such that $x^*\in \bigcap_{i\in S} A_i$. We need to show for the distribution $p=\text{\mps{\frac{1}{3}}}(\mathcal A)$ that 
	\begin{align*}
		\sum_{i\in S} u_i(p)\geq \frac{1}{\frac{3}{2}(1-3^{-t})} \cdot \frac{|S|^2}{n}.
	\end{align*}
	
	To this end, we recall that we have shown in the proof of \Cref{thm:MPS} that every \mps{\gamma} rule with $\gamma\geq \frac{1}{3}$ satisfies that $\sum_{i\in S} u_i(q)\geq\pi(\ell,1) \cdot \frac{|S|^2}{n}$, where $\pi$ denotes the payment willingness function of the corresponding distribution rule, $q$ the distribution chosen by the rule, and $\ell=\max_{i\in N} |A_i|\leq t$ the maximal ballot size in $\mathcal{A}$. By applying this inequality to \mps{\frac{1}{3}} and using that $\pi(\ell,1)\geq \pi(t,1)$, it follows that 
	\begin{align*}
		\sum_{i\in S} u_i(p)\geq \frac{{3^{-1}}}{\sum_{i\in [t]} 3^{-i}} \cdot \frac{|S|^2}{n}=\frac{1}{\sum_{i=0}^{t-1} 3^{-i}}\cdot \frac{|S|^2}{n}.
	\end{align*}
	
	Finally, since $\sum_{j=0}^{t-1} 3^{-j} =\frac{3}{2}(1-3^{-t})$, \mps{\frac{1}{3}} indeed satisfies $\frac{3}{2}(1-3^{-t})$-AFS for profiles with a maximal ballot size of $t$. 
\end{proof}

\end{document}